\documentclass[a4paper,UKenglish,cleveref,autoref,thm-restate]{article}

\usepackage{hyperref}
\usepackage{graphicx} 
\graphicspath{{./figs/}}
\usepackage{amsthm,amsmath,amssymb,amsfonts}
\usepackage{tikz-cd}
\usepackage[normalem]{ulem} 
\usepackage{mathtools} 
\usepackage{fullpage}
\usepackage{comment}
\usepackage{cleveref}
\usepackage{algorithm}
\usepackage[noend]{algpseudocode}
\algnewcommand{\LineComment}[1]{\Statex\(\triangleright\) #1}
\algnewcommand\algorithmicforeach{\textbf{for each}}
\algdef{S}[FOR]{ForEach}[1]{\algorithmicforeach\ #1\ \algorithmicdo}

\newcommand {\mm}[1] {\ifmmode{#1}\else{\mbox{\(#1\)}}\fi}

\newcommand{\Rspace}        {\mm{\mathbb{R}}}

\newcommand{\har}           {\mm{\mathbb{H}}}

\newcommand{\eps}        {\varepsilon}

\newcommand{\Rcal}        {\mm{\mathcal{R}}}

\newcommand{\etal}{{et al.}}

\newcommand{\hcd}{{harmonic chain barcode}}

\newcommand{\hcb}{{harmonic chain barcode}}
\newcommand{\hcbs}{{harmonic chain barcodes}}

\newcommand{\para}[1]{\vspace{2mm}\noindent{\textbf{\sffamily{#1}}}\hspace{4pt}}

\theoremstyle{definition}
\newtheorem{algr}[algorithm]{Algorithm}

\newcommand{\Dgm}{\mm{\mathrm{Dgm}}}

\newcommand{\lrarrowsp}[1]{\xleftrightarrow{\lbarrowspace#1\lbarrowspace}}

\newcommand{\Mcal}{\mathcal{M}}
\newcommand{\barc}{\mathsf{B}}

\newcommand{\filt}{F}
\newcommand{\filtcnt}{m}
\newcommand{\fsimp}{\sigma}
\newcommand{\lbarrowspace}{\;}
\newcommand{\incto}{\hookrightarrow}
\newcommand{\inctosp}[1]{\xhookrightarrow{\lbarrowspace#1\lbarrowspace}}
\newcommand{\bakincto}{\hookleftarrow}

\newcommand{\hbarc}{\mathsf{B}^\har}
\newcommand{\obarc}{\breve{\mathsf{B}}^\har}

\newcommand{\Hm}{H}
\newcommand{\Chn}{C}
\newcommand{\Zyc}{Z}
\newcommand{\Bnd}{B}
\newcommand{\Real}{\mathbb{R}}
\newcommand{\Cobmat}{\mathtt{Cob}}
\newcommand{\Chnmat}{\mathtt{SC}}

\newcommand{\Bndmat}{\mathtt{BoC}}
\newcommand{\Harmat}{\mathtt{H}}
\newcommand{\Rmat}{\mathtt{R}}
\newcommand{\pivot}{\mathrm{pivot}}

\newcommand{\aG}{\alpha}
\newcommand{\sG}{\sigma}
\newcommand{\iG}{\iota}
\newcommand{\dG}{\delta}

\newcommand{\lG}{\lambda}
\newcommand{\mG}{\mu}
\newcommand{\zG}{\zeta}

\let\bar\overline

\usepackage{amsmath}

\DeclareMathOperator*{\argmin}{arg\,min}

\title{Tracking the Persistence of Harmonic Chains: Barcode and Stability}

\author{Tao Hou\thanks{University of Oregon, \texttt{taohou@uoregon.edu}}, Salman Parsa\thanks{DePaul University, \texttt{s.parsa@depaul.edu}}, Bei Wang\thanks{University of Utah, \texttt{beiwang@sci.utah.edu}}}

\crefname{equation}{Eq.}{Eqs.}
\crefname{figure}{Fig.}{Figs.}
\crefname{tabular}{Tab.}{Tabs.}
\crefname{section}{Sec.}{Secs.}
\crefname{appendix}{App.}{Apps.}
\newtheorem{theorem}{Theorem}

\newtheorem{lemma}[theorem]{Lemma}

\newtheorem{definition}[theorem]{Definition}
\newtheorem{proposition}[theorem]{Proposition}
\newtheorem{corollary}[theorem]{Corollary}

\theoremstyle{remark}
\newtheorem{remark}[theorem]{Remark}
\newtheorem*{example*}{Example}
\begin{document}

\maketitle

\begin{abstract}
The persistence barcode is a topological descriptor of data that plays a fundamental role in topological data analysis. 
Given a filtration of data, the persistence barcode tracks the evolution of its homology groups. In this paper, we introduce a new type of barcode, called the harmonic chain barcode, which tracks the evolution of harmonic chains. 
In addition, we show that the harmonic chain barcode is stable. 
Given a filtration of a simplicial complex of size $m$, we present an algorithm to compute its harmonic chain barcode in $O(m^3)$ time.
Consequently, the harmonic chain barcode can enrich the family of topological descriptors in applications where a persistence barcode is applicable, such as feature vectorization and machine learning.

\end{abstract}

\maketitle

\section{Introduction}
\label{sec:introduction}

There are two primary tasks in topological data analysis (TDA)~\cite{DeyWang2022,EdelsbrunnerHarer2010,Zomorodian2005}: reconstruction and inference. 
In a typical TDA pipeline, the data is given as a point cloud in $\Rspace^N$.  
A ``geometric shape'' $K$ in $\Rspace^N$ is reconstructed from the point cloud, usually as a simplicial complex, and $K$ is taken to represent the (unknown) space $X$ from which the data is sampled. 
$X$ is then studied using $K$ as a surrogate for properties that are invariant under invertible mappings (technically, homeomorphisms). 
The deduced ``topological shape'' is not specific to the complex $K$ or the space $X$, but is a feature of the homeomorphism type of $K$ or $X$. 
For example, a standard round circle has the same topological shape as any closed loop such as a knot in the Euclidean space. 
Although topological properties of $X$ alone are not sufficient to reconstruct $X$ exactly, they are among a few global features of $X$ that can be inferred from the data sampled from $X$. 

An (ordinary) persistence barcode~\cite{CarlssonZomorodianCollins2004,Ghrist2008} (or equivalently, a persistence diagram~\cite{ChazalCohenSteinerGlisse2009, Cohen-SteinerEdelsbrunnerHarer2005, EdelsbrunnerLetscherZomorodian2002}) captures the evolution of homological features in a filtration constructed from a simplicial complex $K$. 
It consists of a multi-set of intervals in the extended real line, where the start and end points of an interval (i.e.,~a bar) are the birth and death times of a homological feature in the filtration. 
Equivalently, a persistence diagram is a multi-set of points in the extended plane, where a point in the persistence diagram encodes the birth and death time of a homological feature. 

A main drawback of ordinary persistent homology is that it does not provide canonical choices of geometric representatives for each bar in the barcode. 
Specifically, there are two levels of choices in assigning representatives for bars in a persistence barcode: first, we choose a basis for persistent homology in a consistent way across the filtration; second, we choose a cycle representative inside each homology class in the basis. Both levels of choices involve choosing among significantly distinct geometric features of data to represent the same bar in the barcode. Since such choices are not canonical, it can be challenging to interpret their underlying geometric meaning. In addition, while existing works~\cite{Oleksiy2010,DeyHouMandal2020, Obayashi2018} compute optimal representatives for ordinary persistence, the optimality would rely on a choice of weights for the simplices and on the optimality criterion, leading to different shapes. Moreover, computing the optimal representatives is NP-hard in many interesting situations  ~\cite{ChenFreedman2011, chambers2019minimum, Chambers2009, Borradaile2020, Grochow2018, DeyHouMandal2020}.

As our first contribution, we introduce a new type of barcode called  \emph{\hcb}. This barcode is obtained by tracking the birth and death of harmonic chains in an increasing filtration.  
Since there is always a \emph{unique} harmonic chain in a homology class, the \emph{harmonic representatives} in our {\hcb} immediately remove the second level of choices in a natural way. The main idea is to take the harmonic chain groups along an increasing filtration, and to observe that the groups grow or shrink due to the birth and death of harmonic chains, resulting in a zigzag module with inclusion maps. 
Like any zigzag module~\cite{CarlssonSilva2010}, the module of harmonic chain groups decomposes into interval modules, which then form the {\hcd}.
Moreover, since the maps in this zigzag module are essentially inclusions,
each representative for a bar in our {\hcb} consists of a single harmonic chain which is alive over the entire bar. We emphasize that our {\hcb} is distinct from the ordinary persistence barcode; see for example \Cref{fig:filtbarc,fig:twobc}. 

The ordinary persistence barcode is shown to be stable~\cite{Cohen-SteinerEdelsbrunnerHarer2005}, which is crucial for applications. The stability means that small changes in the data imply only small changes in the barcode. For our second contribution, we show that our {\hcb} is also stable in the same sense as the stability of persistence barcode~\cite{Cohen-SteinerEdelsbrunnerHarer2005}. 

Finally, we present an algorithm for computing 
the {\hcb} in $O(m^3)$ time for a filtration of size $m$,
matching the complexity of practical algorithms for ordinary persistence.

The interpretation of homological features is crucial for applications such as extracting hierarchical structures of amorphous solids~\cite{hiraoka2016hierarchical} and quantifying the growing branching architectures of plants~\cite{LiDuncanTopp2017}.
Such an interpretation is closely related to the existence of canonical chain representatives for a barcode. 
In our harmonic chain barcode, using orthogonality, each bar enjoys a canonical choice of representative (within a homology class) which \emph{lives exactly at the time-interval of the bar}. This arguably promises a more interpretable data feature. Therefore, we expect our {\hcbs} to enrich the family of topological descriptors in applications where ordinary persistence barcode is used, such as feature vectorization and machine learning. Note that just from the fact that our barcode is different from the ordinary persistence one cannot argue in favor or against our suggested barcode. Such a comparison is only meaningful in a specific domain of application and when performed using experimental evaluations. We refer to~\cite{GurnariGuzman-SaenzUtro2023} for an interesting pipeline of data analysis using extra properties of harmonic chains. \cite{GurnariGuzman-SaenzUtro2023} also contains a method of propagating the information from harmonic chains back to original data points which we find also relevant to our harmonic representatives.


\section{Related Work}
\label{sec:related-work}

Harmonic chains were first studied in the context of functions on graphs, where they were identified as the kernel of the Laplacian operator on graphs~\cite{Kirchhoff1847}. 
The graph Laplacian and its kernel are important tools in studying graph properties, see~\cite{Merris1994,MoharAlaviChartrand1991} for surveys. 
Eckmann~\cite{Eckmann1944} introduced the higher-order Laplacian for simplicial complexes, and proved the isomorphism of harmonic chains and homology. 
Guglielmi \etal~\cite{GuglielmiSavostianovTudisco2023} studied the stability of higher-order Laplacians. Horak and Jost~\cite{HorakJost2013} defined a weighted Laplacian for simplicial complexes. Already their theoretical results on Laplacian~\cite{HorakJost2013} anticipated the possibility of applications, as the harmonic chains are thought to contain important geometric information.
This has been validated by recent results that use curves of eigenvalues of Laplacians in a filtration in data analysis~\cite{ChenQiuWang2022, WangNguyenWei2020}. 
The Laplacian was applied to improve the mapper algorithm~\cite{MikePerea2019}, and for coarsening triangular meshes~\cite{KerosSubr2023}. 
The persistent Laplacian~\cite{MemoliWanWang2022} and its stability~\cite{LiuLiWu2023} is an active research area. 
Due to the close relation of harmonic chains and Laplacians, harmonic chains could find applications in areas that Laplacians have been used.
Computing reasonable representative cycles for persistent homology is also an active area of research. 
Here, usually an optimality criterion is imposed on cycles in a homology class to obtain a unique representative. For a single homology class, a number of works~\cite{Chen2010,DeyHiraniKrishnamoorthy2010, DeyHouMandal2020, Obayashi2018,Oleksiy2010, ChambersParsaSchreiber2022, ChenFreedman2011, DeyHiraniKrishnamoorthy2010, LeeChungChoi2019} consider different criteria for optimality of cycles. Hardness of computing optimal representatives has been studied by~\cite{ChenFreedman2011, chambers2019minimum, Chambers2009, ChambersParsaSchreiber2022, Borradaile2020, Grochow2018, DeyHouMandal2020}. For persistent homology, Dey et al.~\cite{DeyHouMandal2020} studied the hardness of choosing optimal cycles for persistence bars. 
Furthermore, De Gregorio et al.~\cite{DeGregorioGuerraScaramuccia2021} used harmonic cycles in a persistent homology setting to compute the persistence barcode. Lieutier~\cite{Lieutier} studied the harmonic chains in persistent homology classes, called persistent harmonic forms.

\para{Relation to the work of Basu and Cox.}~The most relevant work to ours is the inspiring work of Basu and Cox~\cite{BasuCox2022}. 
Basu and Cox had a similar goal as ours, namely, to associate geometric information to each bar in order to obtain a more interpretable data feature. To that end, they introduced the notion of \emph{harmonic persistent barcode}, by associating a subspace of harmonic chains to each bar in the \textit{ordinary} persistence barcode. When the multiplicity of the bar is 1, the subspace is 1-dimensional. 
This is the space of harmonic chains that are born at $s$ and die entering $t$. In general, this is a quotient subspace of harmonic chains at time $s$. Using orthogonality 
they can represent this subquotient as a subspace of the harmonic cycles at time $s$. As a result, they successfully assign a canonical harmonic cycle to a bar in the ordinary persistence diagram. We note that, in general, a bar in the persistence barcode cannot be represented using a single harmonic chain that \textit{remains harmonic during the lifetime of the bar}. The harmonic cycle associated to a bar using the Basu-Cox approach is the initial harmonic representative of the bar, that is, the harmonic cycle that represents the bar at its birth. Basu and Cox also used the \emph{terminal harmonic cycle} in some of their arguments, thus showing that the choice is not entirely canonical. However, Basu and Cox proved significant properties of the initial cycles in terms of what they call \emph{relative essential content}. The novelty of our result is that, in contrast to~\cite{BasuCox2022}, we define a new barcode distinct from the ordinary persistence barcode, in which each bar has a harmonic cycle associated with it which is harmonic during the lifetime of the bar. 
Basu and Cox also proved stability for their harmonic persistent barcode, by considering subspaces as points of a Grassmannian manifold and measuring distances in the Grassmannian. Such a distance quantifies the angles between subspaces, whereas our notions of stability are stronger in the sense that they use the classical bottleneck distance analogous to the ordinary persistence homology.
G\"{u}len et al.~\cite{GulenMemoliWan2024} studied a method that permits the construction of stable persistence diagrams that are equipped with a canonical choice of representative cycles for filtrations over arbitrary finite posets. If the underlying poset is a finite subset of the real line, using persistent Laplacians, they obtained harmonic cycles connected (through a certain isomorphism) to those identified by Basu and Cox.

\section{Background}
\label{sec:background}

In this section, we review the notion of harmonic chains and persistent homology. 
Homology and cohomology are defined with  real coefficients $\mathbb{R}$ (instead of $\mathbb{Z}_2$);
see~\cref{sec:homology} for a review of 
homology and cohomology.
Let $K$ be a simplicial complex and $p$ the homology dimension (or equivalently, homology degree). 
$C_p(K)$, $Z_p(K)$, and $H_p(K)$ denote the $p$-th chain group, cycle group, and homology group of $K$, whereas $C^p(K)$, $Z^p(K)$, and $H^p(K)$ denote the $p$-th cochain group, cocycle group, and cohomology group of $K$, respectively.
Groups across all dimensions are denoted as $C_*(K)$, $C^*(K)$, etc. 
We use $\partial$ and $\delta$ to denote boundary and coboundary operators, respectively. 

\subsection{Harmonic Cycles} 
Based on the standard notions of homology and cohomology  with $\Rspace$ coefficients (see~\cref{sec:homology}), we identify chains and cochains via duality, i.e.,~$C_p(K)= C^p(K)=\Rspace^{n_p}$, where $n_p$ is the number of $p$-simplices. 
Therefore, we can talk about coboundaries of cycles in $Z_p(K)$. 
We first introduce the notion of harmonic chains.  
\begin{definition}
\label{def:harmonic-cycle}
The \emph{$p$-th harmonic chain group} of $K$, denoted $\har_p(K)$, is the group of $p$-cycles that are also $p$-cocylces. 
Equivalently, $\har_p(K) := Z_p(K)\cap Z^p(K)$.
Each element in $\har_p(K)$ is called a \emph{harmonic $p$-chain}.
The harmonic chain group in all dimensions is the group $\har(K):=\bigoplus_p \har_p(K)$.
\end{definition}
We sometimes use \emph{harmonic cycles} in place of \emph{harmonic chains} to emphasize the fact that harmonic chains are cycles.

\begin{lemma}[\cite{Eckmann1944}]
\label{lem:har-homo-iso}
$\har_p(K)$ is isomorphic to $H_p(K)$ and $H^p(K)$. Specifically, each homology and cohomology class has a unique harmonic cycle in it.
\end{lemma}

Harmonic cycles enjoy certain geometric properties. As an example, we mention the following \cref{prop:uniqueness}; see~\cite[Proposition 3]{DeSilvaMorozovVejdemoJohansson2011} for a proof. In other words, a harmonic cycle is the chain with the least squared-norm in a cohomology class.

\begin{proposition}[\cite{DeSilvaMorozovVejdemoJohansson2011}]
\label{prop:uniqueness}
Let $\alpha \in C^p(K)$ be a cochain. There is a unique solution $\bar{\alpha}$ to the least-squares minimization problem 
$\argmin_{\bar{\alpha}} \{ || \bar{\alpha}||^2 \mid \exists \gamma \in C^{p-1}(K);  \alpha = \bar{\alpha}+\delta \gamma \}.$ 
Moreover, $\bar{\alpha}$ is characterized by the relation $\partial {\bar{\alpha}}=0$.
\end{proposition}

There is an alternative definition of harmonic cycles. Consider the natural inner product on $C_p(K)$ given by $\langle\sigma_i , \sigma_j\rangle = \delta_{i,j}$. The harmonic chain group can be defined as $\har_p(K) = Z_p(K) \cap B_p(K)^\perp$. With this definition, the isomorphism of~\cref{lem:har-homo-iso} is realized by a  map that sends $z+B_p(K)$ to its projection to $B_p(K)^\perp$. In addition, the harmonic cycles satisfy 
$\har_p(K) = \ker(\partial_p) \cap \ker(\partial_{p+1}^\perp).$
For a short proof of the equivalence among the definitions of harmonic cycles above, see~\cite{BasuCox2022}. 
Importantly, our algorithm relies on the fact that harmonic cycles are cocycles whose boundaries are zero. 

\subsection{Ordinary Persistence}
\label{sec:persistence}

Ordinary persistent homology takes a \emph{filtration} of a simplicial complex $K$ as input. A \emph{continuous}  filtration $F$  assigns to each $r\in \Rspace$ a subcomplex $K_r \subseteq K$ such that  $K_r \subseteq K_s$ for $r \leq s$. 
Since $K$ is finite, there are finitely many  $t_0,t_1, \ldots, t_m \in \Rspace$ where $K_{t_i}$ changes. 
 We then abuse the notations slightly by
letting $K_i := K_{t_i}$ and have a \emph{discrete} form of $F$,
\begin{equation}
\label{eq:fcomplex}
    \filt:\emptyset = K_0 \hookrightarrow K_1 \hookrightarrow \cdots \hookrightarrow K_{m-1} \hookrightarrow K_m=K,
\end{equation}
where
each $K_i \hookrightarrow K_{i+1}$ is an inclusion. 
Unless stated otherwise, we assume that $\filt$ is \emph{simplex-wise}, i.e., each 
two 
$K_i$ and $K_{i+1}$ differ by at most a single simplex. 
For simplicity of the exposition,
complexes in $\filt$ sometimes are
subscripted by real-valued ``timestamps'' of the form $K_{t_i}$ (e.g., in~\cref{sec:stability}) or subscripted by integers of the form $K_i$ (e.g., in~\cref{sec:algorithm}), which  should not cause any confusions.
Applying homology functor to~\Cref{eq:fcomplex}, we obtain a sequence of homology groups and connecting linear maps (homomorphisms), forming a \emph{persistence module}:
\begin{equation}\label{eq:fhomology}
\Mcal: H_p(K_0) \to H_p(K_1) \to  \cdots \to H_p(K_{m-1}) \to H_p(K_m).   
\end{equation} 
For $s\leq t$, let $f^{s,t}_p: H_p(K_s) \rightarrow H_p(K_t)$ denote the map induced on the $p$-th homology by  inclusion.
The image of the map, $f^{s,t}_p(H_p(K_s)) \subseteq H_p(K_t)$, is called the $p$-th \emph{$(s,t)$-persistent homology group}, denoted $H_p^{s,t}$. The group $H_p^{s,t}$
consists of homology classes which exist in $K_s$ and survive until $K_t$. The dimensions of these vector spaces are the \emph{persistent Betti numbers}, denoted $\beta^{s,t}_p$.
An \emph{interval module}, denoted $I=I[b,d)$, is a persistence module of the form
\[ 0 \rightarrow \cdots \rightarrow 0 \rightarrow \Rspace \rightarrow \cdots \rightarrow \Rspace \rightarrow 0 \rightarrow \cdots \rightarrow 0.\]
In the above, $\Rspace$ is generated by a homology class and the connecting homomorphisms map generator to generator. 
We have $I_r = \Rspace$
for $b \leq r < d$ and $I_r=0$ for other $r$. 
Any persistence module can be decomposed into a collection of interval modules in a unique way \cite{LuoHenselman-Petrusek2023}. 
The collection of $[b,d)$ for all interval modules is called the \emph{persistence barcode}. 
When plotted as points in an extended plane, the result is the equivalent \emph{persistence diagram}. 
 
\para{Stability of Persistence Diagram/Barcode.}
The stability of persistence diagrams (or barcodes) is a crucial property for applications. It says that small changes in data lead to small changes in the persistence diagrams. 
We only review the stability for sublevel-set filtrations here; see \cref{subsec:ordinaryinterleaving} for more on stability of persistence modules.
Let $D$ and $D'$ denote two persistence diagrams. Recall that a persistence diagram is a multi-set of points in the extended plane (each of which is a birth-death pair) which also contains all points on the diagonal.
The \emph{bottleneck distance} of $D,D'$ is defined as 
\[
d_B(D, D') = \text{inf}_\gamma \text{sup}_{p \in D} || p - \gamma(p) ||_\infty,
\]
where $\gamma$ ranges over all bijections between $D$ and $D'$ and $||\cdot ||_\infty$ is the largest absolute value of differences of the points' coordinates.
  
A function $\tilde{f} : |K| \rightarrow \Rspace$ is called \emph{simplex-wise linear} if it is linear on each simplex.
Let $K_r = \{ \sigma \in K\mid\forall x \in |\sigma|, \tilde{f}(x) \leq r \}$ be the \textit{sublevel set complex} at value $r\in \Rspace.$ The complexes $K_r$ and the inclusions between them form a \textit{sublevel set filtration} $\tilde{F}$ (which is not necessarily simplex-wise). We denote the persistence diagram of $\tilde{F}$ as $\Dgm(\tilde{f})$.
We refer to~\cite{Cohen-SteinerEdelsbrunnerHarer2005, ChazalCohenSteinerGlisse2009} for proof of the following \Cref{theorem:persistence-stability}. See~\cite{Cohen-SteinerEdelsbrunnerHarer2005, ChazalCohenSteinerGlisse2009} for the stability of ordinary persistence.

\begin{theorem}
    Let $\tilde{f}, \tilde{g}: |K| \xrightarrow{} \mathbb{R}$ be simplex-wise linear functions. Then $$ d_B( \Dgm(\tilde{f}), \Dgm(\tilde{g})) \leq ||\tilde{f} - \tilde{g} ||_\infty.$$
\end{theorem}

\subsection{Zigzag Persistence}
\label{sec:zigzag}

We provide a brief overview of zigzag persistence; see~\cite{CarlssonSilva2010,CarlssonSilvaMorozov2009} for details.
A zigzag module 
\[\Mcal:V_0\lrarrowsp{g_0}V_1\lrarrowsp{g_1}\cdots\lrarrowsp{g_{k-1}}V_k\]
is a sequence of vector spaces connected by linear maps which could be forward or backward,
i.e., each $g_i$ could be $g_i:V_i\to V_{i+1}$ or $g_i:V_{i}\leftarrow V_{i+1}$.
The module $\Mcal$ decomposes into a direct sum of interval modules $I[b,d]$ of the form 
\[0 \longleftrightarrow \cdots \longleftrightarrow 0 \longleftrightarrow \Rspace \longleftrightarrow \cdots \longleftrightarrow \Rspace \longleftrightarrow 0 \cdots \longleftrightarrow 0 \]
with 1-dimensional vector spaces in the range $[b,d]$. 
The multi-set of intervals in the decomposition
defines the \emph{barcode} of $\Mcal$, denoted as $\barc(\Mcal)$.

\para{Conventions.}
In this paper, we may omit the subscript/dimension of a homology group if there is no danger of ambiguity. Moreover, we use the terms persistence barcode and persistence diagram interchangeably.

\section{Harmonic Chain Barcodes and Representatives}
\label{sec:barcode-zigzag}

As reviewed in \cref{sec:persistence}, by directly taking the homology functor, an increasing filtration of simplicial complexes leads to an \emph{ordinary persistence module} consisting of homology groups~\cite{chazal2016structure,DeyWang2022}. These homology groups are connected by forward maps of the form $H(K_i) \xrightarrow{} H(K_{i+1})$. The ordinary persistence module then decomposes into interval  modules, which define the \emph{ordinary persistence barcode}. 
In this section, we show that the harmonic chain groups of all the complexes in a filtration constitute an abstract \emph{zigzag persistence module}  (see~\cref{sec:zigzag}). 
The main idea is that we take the harmonic chain groups along the filtration, where the groups could grow or shrink, resulting in a zigzag module. 
Moreover, the maps in this zigzag module are inclusions. 
Like any zigzag module~\cite{CarlssonSilva2010}, the module of harmonic chain groups decomposes into interval modules. These intervals form the {\hcd}, the main object of interest in this paper. 

Throughout the  section,
consider a simplex-wise filtration
\begin{equation*}
\filt: 
\emptyset
=K_0 \inctosp{\fsimp_0} K_1 \inctosp{\fsimp_1} 
\cdots \inctosp{\fsimp_{\filtcnt-1}} K_\filtcnt=K,
\end{equation*}
where each $K_{i+1}$ differs from $K_i$
by the addition of a simplex $\fsimp_i$. 
Recall that $K_i := K_{t_i}$ for $t_i\in\Real$ where $K_{t_i}$
is a complex from a \emph{continuous} filtration indexed over $\Real$; see \Cref{sec:persistence}. 


\begin{proposition}
\label{prop:arrow-dim-change}
For each inclusion $K_i \inctosp{\fsimp_i} K_{i+1}$ in $\filt$: 
\[\dim(\har(K_{i+1}))=\dim(\har(K_{i})) \pm 1.\]
\end{proposition}
\begin{proof}
This follows from \Cref{lem:har-homo-iso} and some well-known facts in persistence (see~\cite{DeyWang2022,EdelsbrunnerLetscherZomorodian2002}).
\end{proof}

\begin{definition}
\label{def:pos-neg-simp}
A simplex $\fsimp_i$ inserted in $\filt$ is called \emph{positive}
if $\dim(\har(K_{i+1}))=\dim(\har(K_{i}))+1,$ and \emph{negative}
if $\dim(\har(K_{i+1}))=\dim(\har(K_{i}))-1$.
\end{definition}

We describe how we connect the harmonic chain groups of complexes in $\filt$ by inclusions.
For each $K_i$ and each chain $c=\sum_{j=0}^{i-1}\aG_j\fsimp_j$ in $\Chn_p(K_i)$,
we identify $c$ as a chain $c=\sum_{j=0}^{\filtcnt-1}\aG_j\fsimp_j$ in $\Chn_p(K)$,
where $\aG_j=0$ for $j\geq i$.
This makes 
both 
$\Chn_p(K_i)$ and $\Chn_p(K_{i+1})$ a subspace of $\Chn_p(K)$. 
We then have the following inclusion:
\begin{equation}
\label{eqn:zyc-inc}
\Zyc_p(K_i) \incto \Zyc_p(K_{i+1}).
\end{equation}
Recall that $\har_p(K_i):=\Zyc_p(K_i)\cap \Zyc^p(K_i)$, which means that
$\har_p(K_i)\subseteq \Zyc_p(K_i) \subseteq \Chn_p(K)$.
Hence, we also identify each harmonic cycle in $\har_p(K_i)$ as a chain in $\Chn_p(K)$. 
We then observe in \cref{thm:harmonic-inclusion}
 a similar inclusion as in \Cref{eqn:zyc-inc} between any harmonic chain groups $\har_p(K_i)$ and $\har_p(K_{i+1})$,
with a possible flip on the direction.
\begin{theorem}
\label{thm:harmonic-inclusion}
For each arrow $K_i \inctosp{\fsimp_i} K_{i+1}$ in $\filt$, where $\fsimp_i$ is 
a $p$-simplex:  
\begin{itemize}
\item There is an inclusion $\har_p(K_i)\incto \har_p(K_{i+1})$
if $\fsimp_i$ is positive;
\item And there is an inclusion $\har_{p-1}(K_i)\bakincto\har_{p-1}(K_{i+1})$ if $\fsimp_i$ is negative.
\end{itemize}
In addition, in each case the harmonic chain groups in other dimensions remain unchanged, i.e., $\har_q(K_i) = \har_q(K_{i+1})$ for any other $q \notin \{p, p-1\}$.
\end{theorem}
\begin{proof}
The only harmonic chain groups that might change from $K_i$ to $K_{i+1}$ are 
those in dimension $p$ because cycle and cocycle groups in other dimensions  do not change.

First consider Case I that $\fsimp_i$ is positive.
Let $\iota: K_i \incto K_{i+1}$ be the inclusion map.  As noted above, we identify any $c \in \Chn_*(K_i)=\Chn^*(K_i)$ with $\iota_\sharp (c) \in \Chn_*(K_{i+1})=\Chn^*(K_{i+1})$.

\begin{description}
\item[Case I.1: $\har_p(K_i)\subseteq \har_p(K_{i+1})$.] 
Take $z \in \har_p(K_{i})$. 
Obviously, $\iota_\sharp(z)$ is a cycle in $K_{i+1}$.
Moreover, for any $c \in C_{p+1} (K_{i+1})$, $\delta (\iota_\sharp (z)) (c) = \iota_\sharp(z) (\partial c)$. Since $C_{p+1} (K_i)= C_{p+1} (K_{i+1})$, $c$ and hence $\partial c$ exist in $K_i$, meaning that $\iota_\sharp(z) (\partial c) = z (\partial c) = (\delta z) (c) = 0$. It follows that $\iota_\sharp(z)$ is a cocycle in $K_{i+1}$. Therefore, $\har_p(K_i) \subseteq \har_p (K_{i+1})$.

\item[Case I.2: $\har_{p-1}(K_i)= \har_{p-1}(K_{i+1})$.] 
Take $z \in \har_{p-1}(K_{i+1})$. 
Since $\fsimp_i$ is a positive $p$-simplex,
we have $Z_{p-1}(K_{i+1}) = Z_{p-1}(K_i)$
and $B_{p-1}(K_{i})=B_{p-1}(K_{i+1})$.
So $z \in Z_{p-1}(K_i)$
and we consider $z$ as a cochain in $K_i$.
Then, for any $c \in C_{p}(K_{i})$,  $\delta (z) (c) = z(\partial c) = 0$, with the last equality due to 
$\partial c\in B_{p-1}(K_{i})=B_{p-1}(K_{i+1})$.
Therefore, $\har_p(K_{i+1}) \subseteq \har_p(K_i)$.
Take $z \in \har_{p-1}(K_{i})$. For any $c \in C_{p}(K_{i+1})$,  $\delta (\iota_\sharp z) (c) = (\iota_\sharp z)(\partial c) = 0$, 
with the last equality due to 
$\partial c\in B_{p-1}(K_{i+1})=B_{p-1}(K_{i})$.
Therefore, $\har_{p-1}(K_i)= \har_{p-1}(K_{i+1})$.
\end{description}
    
Now consider Case II that $\fsimp_i$ is negative.

\begin{description}
\item[Case II.1: $\har_p(K_{i+1}) = \har_p(K_i)$.] 
Since $\fsimp_i$ is negative, $Z_p(K_{i+1}) = Z_p(K_i)$
and $B_p(K_{i+1})=B_p(K_{i})$.
The verification for this case is then the same as the verification
for Case I.2 with a shift on the homology degree.

\item[Case II.2: $\har_{p-1}(K_{i+1}) \subseteq \har_{p-1}(K_i)$.]
Take $z \in \har_{p-1}(K_{i+1})$. We have $Z_{p-1}(K_{i+1}) = Z_{p-1}(K_i)$.
Therefore, $z \in Z_{p-1}(K_i)$
and we consider $z$ as a cochain in $K_i$. 
For any $c \in C_{p}(K_{i})$,  $\delta (z) (c) = z(\partial c) = 0$, 
with the last equality due to $\partial c\in B_{p-1}(K_i) \subseteq B_{p-1}(K_{i+1})$.
Therefore, $\har_{p-1}(K_{i+1}) \subseteq \har_{p-1}(K_i)$.\qedhere
\end{description}
\end{proof}

\begin{definition}[Harmonic chain barcode]
Consider the following \emph{harmonic zigzag module}
\begin{equation}
\label{eqn:har-zz-mod}
\har(\filt): \har(K_0) \leftrightarrow \har(K_1) \leftrightarrow
\cdots \leftrightarrow \har(K_\filtcnt),
\end{equation}
where the harmonic chain groups are connected by either forward inclusions {\rm(}e.g., $\har(K_i) \to \har(K_{i+1})${\rm)} or backward inclusions {\rm(}e.g., $\har(K_i) \leftarrow \har(K_{i+1})${\rm);}
see \Cref{thm:harmonic-inclusion}.
Define the \emph{\hcd} $\hbarc(\filt)$ of $\filt$ 
as the barcode of the zigzag module $\har(\filt)$,
that is, \[\hbarc(\filt):=\barc(\har(\filt)).\]
\end{definition}

In general, the {\hcb} is different from the ordinary persistence barcode for a filtration; see~\Cref{fig:twobc}.
In this paper, we also consider the zigzag module $\har_p(\filt)$ derived by taking $\har_p(K_i)$ on each $K_i$ in $\filt$.
While the $p$-th harmonic chain groups in $\har_p(\filt)$ are still connected by forward or backward inclusions, we may have $\har_p(K_i)=\har_p(K_{i+1})$ for two consecutive groups in $\har_p(\filt)$.
We therefore define the \emph{$p$-th {\hcb}} $\hbarc_p(\filt)$ of $\filt$ 
as the barcode of $\har_p(\filt)$,
i.e.,~$\hbarc_p(\filt):=\barc(\har_p(\filt))$.
Since $\har(\filt)=\bigoplus_p\har_p(\filt)$ (following from \Cref{thm:harmonic-inclusion}),
we have $\hbarc(\filt)=\bigsqcup_p\hbarc_p(\filt)$.

\para{Harmonic Representatives.}
In the rest of the section, we define harmonic representatives.
\begin{proposition}
Let $[b,d]$ be an interval in $\hbarc(\filt)$.
The inclusion $\har(K_{b-1}) \incto \har(K_b)$ in $\har(\filt)$
is forward.
Moreover, if $d<\filtcnt$,
then the inclusion $\har(K_{d}) \bakincto \har(K_{d+1})$ is backward.
\end{proposition}
\begin{proof}
This follows from \Cref{thm:harmonic-inclusion} and the fact that 
$\dim(\har(K_{b}))=\dim(\har(K_{b-1}))+1$ 
and
$\dim(\har(K_{d+1}))=\dim(\har(K_{d}))-1$.
\end{proof}

Representatives for general zigzag modules were introduced in~\cite{BendichPaul2013Haro,maria2015zigzag} (see also~\cite{zzrep}),
which consist of a sequence of cycles
for each interval.
However, since the harmonic chain groups are connected by inclusion maps in $\har(\filt)$,
we have:
\begin{proposition}\label{prop:single-rep}
Each representative for an interval in $\hbarc(\filt)$
contains a single harmonic cycle.
\end{proposition}

We then adapt the definition of representatives for
general zigzag modules~\cite{zzrep,maria2015zigzag} and 
define
harmonic representatives as follows:
\begin{definition}[Harmonic representative]\label{dfn:rep}
A \emph{harmonic $p$-representative}
{\rm(}or simply \emph{$p$-representative}{\rm)}
for an interval $[b,d]\in\hbarc_p(\filt)$ is a $p$-cycle 
$z\in\har_p(K_i)$ for $i\in[b,d]$ satisfying:
\begin{description}
\item[Birth condition:] $z$ is born in $\har_p(K_{b})$, i.e., $z\in\har_p(K_{b})\setminus\har_p(K_{b-1})$ {\rm(}notice that $b>0${\rm);}
\item[Death condition:] $z$ dies leaving $\har_p(K_{d})$, i.e., $z\in\har_p(K_{d})\setminus\har_p(K_{d+1})$ if $d<m$.
\end{description}
\end{definition}

Sometimes we relax the restriction of $[b,d]\in\hbarc_p(\filt)$ 
and have a harmonic representative for an arbitrary integer interval $[b,d]\subseteq[0,\filtcnt]$.
\Cref{sec:rep-just} details the original definition of zigzag representatives
and justifies \Cref{prop:single-rep} and \Cref{dfn:rep}. 
\section{A Cubic-Time Algorithm for Computing Harmonic Chain Barcodes}\label{sec:algorithm}
In this section, we propose an $O(\filtcnt^3)$ algorithm for computing the {\hcb}  and its harmonic representatives given a filtration containing $\filtcnt$ insertions.
We first overview the algorithm and then describe the implementation.
While there have been algorithms~\cite{CarlssonSilvaMorozov2009,zzrep,maria2015zigzag} for computing zigzag barcodes in the general case, these algorithms target zigzag modules induced from directly \emph{taking the homology functor on zigzag filtrations}.
In contrast, our algorithm targets the special type of zigzag module $\har(\filt)$, where harmonic chain groups are derived from an \emph{ordinary non-zigzag} filtration and
are connected by inclusions.

For simplicity, when describing the algorithm, complexes in a filtration are always 
indexed by integers instead of the real-valued timestamps (see~\Cref{sec:persistence}).
Again, we assume that the input is a simplex-wise filtration 
$\filt: 
\emptyset
=K_0 \inctosp{\fsimp_0} K_1 \inctosp{\fsimp_1} 
\cdots \inctosp{\fsimp_{\filtcnt-1}} K_\filtcnt=K,
$
where each $K_{i+1}$ differs from $K_i$ by the addition of a simplex $\fsimp_i$. 

\Cref{alg:frame} provides an overview: our algorithm computes the harmonic chain barcode by iteratively maintaining the harmonic representatives. 
It processes the insertions in $\filt$ one by one and finds \emph{pairings} of the \emph{birth indices} (starting points of the intervals) and \emph{death indices} (ending points of the intervals) to form intervals in $\hbarc(\filt)$. 
When we encounter a new birth index, it is initially \emph{unpaired};
when we encounter a new death index, an unpaired birth index is chosen to pair with the death index. 

The following definition helps present the algorithm.
\begin{definition}
For an interval $[b,i] \subseteq [0,\filtcnt]$,
a \emph{partial $p$-representative} for $[b,i]$ is a $p$-representative as in \Cref{dfn:rep}
by ignoring the death condition.
\end{definition}

\makeatletter
  \def\vhrulefill#1{\leavevmode\leaders\hrule\@height#1\hfill \kern\z@}
\makeatother

\newenvironment{algrl}
    {\vspace{5pt}
\par\nobreak\nobreak\noindent\vhrulefill{0.8pt}
\begin{algr}
\vspace{-5pt}
    }
    { 
\vspace{-3pt}    
\par\nobreak\nobreak\noindent\vhrulefill{0.8pt}
\end{algr}
\vspace{3pt}
    }

\begin{algrl}\label{alg:frame}
Let $U^p$ be the set of unpaired birth indices for each homology degree $p$, where $U^p=\emptyset$ initially. Before each iteration that processes the insertion $K_{i}\inctosp{\fsimp_i}K_{i+1}$, we maintain a $p$-cycle $z$ for each $b\in U^p$ that is a partial representative for the interval $[b,i]$.
We use partial representatives to determine a finalized representative when a birth index is paired with a death index (by making sure the death condition is satisfied).

When processing the insertion of a $p$-simplex $\fsimp_i$ via $K_{i}\inctosp{\fsimp_i}K_{i+1}$, we proceed as follows:
\begin{description}
    \item[If $\sG_i$ is positive:] 
    \begin{itemize}
        \item[]
        \item Since $\dim(\har_p(K_{i+1}))=\dim(\har_p(K_{i}))+1$ (by \Cref{thm:harmonic-inclusion}), we add a new birth index $i+1$ to $U^{p}$.
        \item Find a harmonic $p$-cycle $z\in\har_p(K_{i+1})\setminus\har_p(K_{i})$, and let $z$ be the partial representative for $i+1\in U^p$.
    \end{itemize}
    
    \item[If $\sG_i$ is negative:] 
        \begin{itemize}
        \item[]
        \item Since $\dim(\har_{p-1}(K_{i+1}))=\dim(\har_{p-1}(K_{i}))-1$, we have a new death index $i$.
        \item Let $U^{p-1}=\{b_j\mid j=1,\ldots,k\}$, where each $b_j$ has a partial $(p-1)$-representative $z_j$.
        \item Let $U'=\{b_j\in U^{p-1}\mid z_j\not\in\har_{p-1}(K_{i+1})\}$, i.e., $U'$ contains all  birth indices in $U^{p-1}$ whose partial representatives do not persist to $K_{i+1}$.
        \item Pair the \emph{smallest} (i.e., the ``oldest'') index $b_{j^*}\in U'$ with $i$ which forms a new interval $[b_{j^*},i]\in\hbarc_{p-1}(\filt)$. Assign $z_{j^*}$ as the representative for $[b_{j^*},i]\in\hbarc_{p-1}(\filt)$ and remove $i$ from $U^{p-1}$.
        \item Consider each $b_{j}\in U'\setminus\{b_{j^*}\}$. We have that $z_j\not\in\har_{p-1}(K_{i+1})$ because $\dG(z_j)$ becomes non-zero in $K_{i+1}$ ($\partial z_j$ is always zero after being born). Also, the only reason that  $\dG(z_j)\neq 0$ in $K_{i+1}$ is because  $\aG_j:=\dG(z_j)(\fsimp_i)=z_j(\partial\fsimp_i)\neq 0$. Let $\aG_*=z_{j^*}(\partial\fsimp_i)$. Update the partial representative for $b_j$ as $z_j:= z_j-(\aG_j/\aG_*)\cdot z_{j^*}$ so that $z_j$ now persists to $\har_{p-1}(K_{i+1})$.
\end{itemize}
\end{description}

After processing all the insertions, for each $b$ in each $U^p$ with a partial representative $z$, let $[b,\filtcnt]$ form an interval in $\har_p(\filt)$ with a representative $z$.
\end{algrl}

\begin{example*}
\Cref{fig:filtbarc} provides an example of how \Cref{alg:frame} computes {\hcb} $\har_1(\filt)$ and representatives for a filtration $\filt$, with the resulting barcode  in \Cref{fig:twobc}. For simplicity, $K_0, K_1, K_2$ are omitted and vertex insertions are ignored. 
In $K_4$, $K_5$, and $K_7$, new 1-cycles are born which are also harmonic cycles due to a lack of triangles. The partial representatives $z_4$, $z_5$, $z_7$ for the indices $U^1=\{4,5,7\}$ all persist till $K_7$. 
When the triangle $abe$ is inserted in $K_8$ (producing a death index $7$), $\dG(z_4)(abe)=1$, $\dG(z_5)(abe)=1$, and $\dG(z_7)(abe)=3$.
We pair $4\in U^1$ with $7$ and have $[4,7]\in\har_1(\filt)$ with a representative $z_4$.
We also let $z_5:=z_5-z_4$ and $z_7:=z_7-3z_4$, which now both persist to $\har_1(K_8)$.
The remaining steps are similar. Notice that when $abe$ is inserted, the ``alive'' bar represented by the boundary of $abe$ (i.e.,~$ab+be-ae$) is not killed.
This is a significant difference from ordinary persistence (which will kill the alive bar represented by $\partial(abe)$ when $abe$ is inserted).
\end{example*}

\begin{figure}[!ht]
\centering
\includegraphics[width=\linewidth]{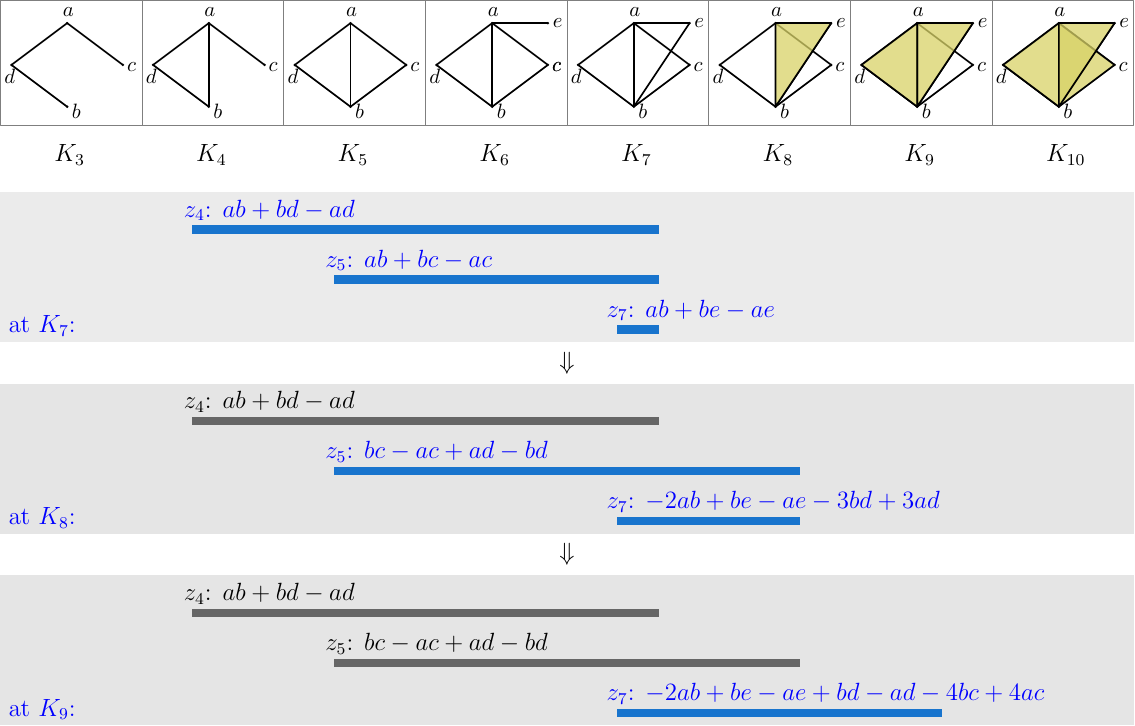}
\caption{An example of the computation of {\hcb} $\har_1(\filt)$ and representatives.}
\label{fig:filtbarc}
\end{figure}

\begin{figure}[!ht]
  \centering
  \includegraphics[width=0.75\linewidth]{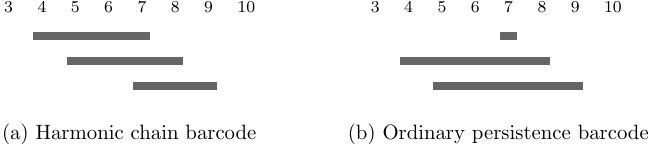}
  \caption{Harmonic chain barcode and ordinary persistence barcode for the filtration in \Cref{fig:filtbarc}. Deviating from conventions in ordinary persistence, bars are drawn as closed integer intervals, e.g., $[7,7]$ in the ordinary barcode is killed by the addition of $abe$ in $K_8$.}
  \label{fig:twobc}
\end{figure}

\begin{remark}
\Cref{alg:frame} chooses the ``oldest'' birth, among the options, to pair with a death while the ordinary persistence chooses the ``youngest'' one~\cite{EdelsbrunnerLetscherZomorodian2002}.
A brief explanation is that a summation of representatives \emph{persists} to the next harmonic chain group in $\har(\filt)$, whereas a summation of representatives
in the ordinary persistence \emph{becomes trivial} in the next homology group (and cannot persist).
We refer to~\cite{zzrep,maria2015zigzag} for a formal definition of the different types of birth/death ends and how they impact the computation of persistence.
\end{remark}

The following (rephrased) proposition from~\cite{dey2021computing} helps prove the correctness of \Cref{alg:frame}:

\begin{proposition}[Proposition 9, \cite{dey2021computing}]
\label{prop:pn-paring-w-rep}
If a pairing $\pi$ of the birth and death indices in $\har(\filt)$ satisfies that each interval from $\pi$ has a harmonic representative, then $\pi$ induces the {\hcb} $\hbarc(\filt)$.
\end{proposition}

\begin{theorem}
\Cref{alg:frame} correctly computes the {\hcb} for filtration $\filt$.
\end{theorem}
\begin{proof}
We first show by induction that before \Cref{alg:frame} processes each $\fsimp_i$, the $p$-cycle $z$ maintained for any $b\in U^p$ is a partial $p$-representative for $[b,i]$.
Suppose that the claim is true before processing $\fsimp_i$.
If $\fsimp_i$ is positive, index $i+1$ is added to $U^p$ and we assign $z \in\har_p(K_{i+1})\setminus\har_p(K_i)$ to $i+1$, which is clearly a partial representative for $[i+1,i+1]$.
For an index $b$ previously in $U^p$ having a partial representative $z$, since $\har_p(K_i)\subseteq \har_p(K_{i+1})$ (\Cref{thm:harmonic-inclusion}), we have that $z$ is still a partial representative for the interval $[b,i+1]$.

In the case that $\fsimp_i$ is negative, whenever we update a representative $z_j$ as $z_j-(\aG_j/\aG_*)\cdot z_{j^*}$, we have that $b_{j^*}<b_j$ and hence the updated $z_j$ remains born in $\har_{p-1}(K_{b_j})$. Therefore, the new $z_j$ is still a partial representative for $[b,i]$, which is also a partial representative for $[b,i+1]$ because now $z_j\in\har_{p-1}(K_{i+1})$.

By \Cref{prop:pn-paring-w-rep}, we only need to show that each interval output by the algorithm admits a harmonic representative.
This follows from the algorithm. 
For example, whenever we output an interval $[b_{j^*},i]\in\hbarc_{p-1}(\filt)$ when $\fsimp_i$ is negative, we have $z_{j^*}\not\in\hbarc_{p-1}(K_{i+1})$, which ensures the death condition.
\end{proof}

\let\mbeg\[
\let\mend\]


\subsection{Implementation}\label{sec:impl}
Before detailing the implementation,
we first have the following fact:

\begin{proposition}\label{prop:dim-chng}
When  a $p$-simplex $\fsimp_i$ is added  from $K_i$ to $K_{i+1}$,
one has:
   \begin{itemize}
   \item If $\fsimp_i$ is positive, then 
   \begin{align*}
       &\dim(\Zyc^p(K_{i+1}))=\dim(\Zyc^p(K_{i}))+1\text{ and}\\
   &\dim(\Bnd^p(K_{i+1}))=\dim(\Bnd^p(K_{i})).
   \end{align*}

   \item If $\fsimp_i$ is negative, then
   \begin{align*}
   &\dim(\Zyc^p(K_{i+1}))=\dim(\Zyc^p(K_{i}))+1\text{, }\\
   &\dim(\Bnd^{p}(K_{i+1}))=\dim(\Bnd^{p}(K_{i}))+1\text{, and}\\
   &\dim(\Zyc^{p-1}(K_{i+1}))=\dim(\Zyc^{p-1}(K_{i}))-1.
   \end{align*}
   \end{itemize}
Furthermore, in each case, all other cocycle or coboundary  groups not mentioned above stay the same
from $K_i$ to $K_{i+1}$.
\end{proposition}
\begin{proof}
We first show that $\dim(\Zyc^p(K_{i+1}))=\dim(\Zyc^p(K_{i}))+1$
whether $\fsimp_i$ is positive or negative.
Let \mbeg\iG^\sharp:\Chn^p(K_{i+1})\to\Chn^p(K_{i})\mend
be the cochain map induced by inclusion.
Clearly, we have  $\iG^\sharp(\Zyc^p(K_{i+1}))\subseteq\Zyc^p(K_{i})$~\cite{hatcher2002algebraic}.
Let \mbeg\jmath^\sharp:\Zyc^p(K_{i+1})\to\Zyc^p(K_{i})\mend be the restriction of $\iG^\sharp$.
We show that $\jmath^\sharp$ is surjective
and the kernel of $\jmath^\sharp$ is 1-dimensional.
First, for any cocycle $\phi\in\Zyc^p(K_{i})$,
we define a cocycle $\phi'\in\Zyc^p(K_{i+1})$
where $\phi'(\fsimp_i)=0$
and values of $\phi'$ on other $p$-simplices are
the same as $\phi$.
We have that $\jmath^\sharp(\phi')=\phi$ and so $\jmath^\sharp$
is surjective.
Moreover, let $\hat{\fsimp}_i\in\Zyc^p(K_{i+1})$ be 
the dual of $\fsimp_i$.
We have that $\hat{\fsimp}_i$ generates 
the kernel of $\jmath^\sharp$.

We also notice that the coboundary maps \mbeg\dG^{p-2}: \Chn^{p-2}(K_i)\to \Chn^{p-1}(K_i)\mend
   and \mbeg\dG^{p-2}: \Chn^{p-2}(K_{i+1})\to \Chn^{p-1}(K_{i+1})\mend
   are the same, and so $\Bnd^{p-1}(K_{i+1})=\Bnd^{p-1}(K_{i})$.

Suppose that $\fsimp_i$ is positive.
We have  $\dim(\Hm^p(K_{i+1}))=\dim(\Hm^p(K_{i}))+1$.
Since 
\begin{align*}
&\dim(\Hm^p(K_{i}))=\dim(\Zyc^p(K_{i}))-\dim(\Bnd^p(K_{i})),\\
&\dim(\Hm^p(K_{i+1}))=\dim(\Zyc^p(K_{i+1}))-\dim(\Bnd^p(K_{i+1})),
\text{ and }\\
&\dim(\Zyc^p(K_{i+1}))=\dim(\Zyc^p(K_{i}))+1,
\end{align*}
we have $\dim(\Bnd^p(K_{i+1}))=\dim(\Bnd^p(K_{i}))$.
Similarly, since
\begin{align*}
&\Hm^{p-1}(K_{i+1})\approx\Hm^{p-1}(K_{i}),\\
&\Bnd^{p-1}(K_{i+1})=\Bnd^{p-1}(K_{i}),
\text{ and }\\
&\Chn^{p-1}(K_{i+1})=\Chn^{p-1}(K_{i}),
\end{align*}
we have $\Zyc^{p-1}(K_{i+1})=\Zyc^{p-1}(K_{i})$.

When $\fsimp_i$ is negative,
the statements regarding the dimensions of the 
cocycle and coboundary groups 
in the proposition
can be similarly obtained.

To see that other cocycle or coboundary groups do not change, 
we provide arguments for the groups in degree $p+1$ 
and the justification for groups in other degrees is omitted.
We first have $\Zyc^{p+1}(K_{i+1})=\Zyc^{p+1}(K_{i})$
because the two 
maps $\dG^{p+1}: \Chn^{p+1}(K_i)\to \Chn^{p+2}(K_i)$ and
$\dG^{p+1}: \Chn^{p+1}(K_{i+1})\to \Chn^{p+2}(K_{i+1})$ are the same.
Moreover, 
we have $\Bnd^{p+1}(K_{i+1})=\Bnd^{p+1}(K_{i})$
because $\sG_i$ has no cofaces in $K_{i+1}$.
\end{proof}

\newcommand{\ldim}{{p}}
\para{Pseudo-cochains.}
We introduce the notion of {pseudo-cochains}
which are more ``concise'' versions
of  cochains.
These
pseudo-cochains help maintain a coboundary basis
in our implementation.
Proofs of 
\Cref{prop:cobound-by-bound,prop:pseudocob-is-cob,prop:pseudo-coc-nat}
follow directly from
definitions
and are omitted.
\Cref{prop:cobound-by-bound}
says that 
the coboundary of a $\ldim$-cochain $\phi$
is completely determined by how $\phi$ maps the $\ldim$-boundaries 
to  real values:
\begin{proposition}\label{prop:cobound-by-bound}
Let
$\phi:\Chn_\ldim(K_i)\to\Real$ 
and $\psi:\Chn_\ldim(K_i)\to\Real$ be two $\ldim$-cochains  of $K_i$.
Also,
let $\bar{\phi}:\Bnd_\ldim(K_i)\to\Real$ and $\bar{\psi}:\Bnd_\ldim(K_i)\to\Real$ be restriction of $\phi$ 
and $\psi$
to $\Bnd_\ldim(K_i)$ respectively.
One has that  $\dG(\phi)=\dG(\psi)$
if and only if $\bar{\phi}=\bar{\psi}$.
\end{proposition}

\begin{definition}
Define a \emph{$\ldim$-pseudo-cochain} $\psi$ in $K_i$ as 
a map $\psi:\Bnd_\ldim(K_i)\to\Real$.
Define the \emph{coboundary} of $\psi$
as the $(\ldim+1)$-cochain $\psi\circ\bar{\partial}$,
where $\bar{\partial}:\Chn_{\ldim+1}(K_i)\to\Bnd_\ldim(K_i)$
is the restriction of $\partial$ to its image.
We abuse the notation slightly and denote 
the coboundary of $\psi$ as $\dG(\psi)$.
\end{definition}

\Cref{prop:pseudocob-is-cob} 
says that coboundaries of
a complex  are completely determined by the pseudo-cochains:

\begin{proposition}\label{prop:pseudocob-is-cob}
The space of coboundaries of all 
$\ldim$-pseudo-cochains of $K_i$ equals $\Bnd^{p+1}(K_i)$.
\end{proposition}

\begin{definition}
Given
a basis $\Rcal=\{\zG_1,\ldots,\zG_r\}$ of $\Bnd_\ldim(K_i)$
and
a column vector $\vec{v}=(\aG_1,\ldots,\aG_r)^T$,
we say that  $\vec{v}$ \emph{encodes} 
the $\ldim$-pseudo-cochain $\psi:\Bnd_\ldim(K_i)\to\Real$
w.r.t\ 
the basis $\Rcal$
if $\psi(\zG_j)=\aG_j$ for each $j$.
\end{definition}

\Cref{prop:pseudo-coc-nat} 
says that summations and scalar multiplications
on the column vectors encoding 
$\ldim$-pseudo-cochains
commute with
the coboundary operator:

\begin{proposition}\label{prop:pseudo-coc-nat}
Fix
a basis $\Rcal=\{\zG_1,\ldots,\zG_r\}$ of $\Bnd_\ldim(K_i)$.
Let
$\vec{u}$, $\vec{v}$
be two column vectors encoding 
$\ldim$-pseudo-cochains $\psi$, $\psi'$
w.r.t $\Rcal$ respectively.
For an $\aG\in\Real$,
one has that $\vec{u}+\aG\vec{v}$ encodes
the $\ldim$-pseudo-cochains $\psi+\aG\cdot\psi'$
w.r.t $\Rcal$
and $\dG(\psi+\aG\cdot\psi')=\dG(\psi)+\aG\cdot\dG(\psi')$.
\end{proposition}

In this paper, when a basis
$\Rcal$ of $\Bnd_\ldim(K_i)$ is clear,
we abuse the notations and let $\dG(\vec{v}):=\dG(\psi)$,
where
$\vec{v}$ is a vector
encoding
a $\ldim$-pseudo-cochain $\psi$ w.r.t $\Rcal$.

\para{Representation of chains and cochains.}
We represent
chains   as matrix columns
and make no distinction between a  column vector
and the chain it represents.
Specifically, 
with $\Chn_*(K_i)\subseteq\Chn_*(K)$
and a natural basis $\{\fsimp_0,\fsimp_1,\ldots,\fsimp_{\filtcnt-1}\}$
for $\Chn_*(K)$,
each chain \mbeg c=\sum_{k=0}^{i-1}\aG_k\fsimp_k\in\Chn_*(K_i)\mend
is represented as a column vector  \mbeg(\aG_0,\aG_1,\ldots,\aG_{\filtcnt-1})^T,\mend
where $\aG_k=0$ for each $k\geq i$.
For a matrix $\mathtt{M}$, we denote its $j$-th
column as $\mathtt{M}[j]$.
For a chain $z$ as a column vector,
we  denote its $k$-th entry as $z[k]$.
Hence, the entry in the $k$-th row and $j$-th column 
in $\mathtt{M}$ is denoted as $\mathtt{M}[j][k]$.
Moreover,
we let the \emph{pivot} of $\mathtt{M}[j]$
be the index of the lowest non-zero entry in 
$\mathtt{M}[j]$
which is denoted as $\pivot(\mathtt{M}[j])$.

Cochains are similarly represented as matrix columns
by taking the dual basis \mbeg\{\hat{\fsimp}_0,\hat{\fsimp}_1,\ldots,\hat{\fsimp}_{\filtcnt-1}\}\mend
for $\Chn^*(K)$.
Notice that since chains and cochains
are  identified, 
a matrix column \mbeg(\aG_0,\aG_1,\ldots,\aG_{\filtcnt-1})^T\mend
can represent \emph{both} the chain $\sum_{k=0}^{\filtcnt-1}\aG_k\fsimp_k$ and the cochain
$\sum_{k=0}^{\filtcnt-1}\aG_k\hat{\fsimp}_k$.
Moreover, due to the identification,
 a formal sum \mbeg\sum_{k=0}^{\filtcnt-1}\aG_k\hat{\fsimp}_k\mend
can also denote a chain,
which should not cause any confusions.

\para{Matrices maintained.}
The only step in \Cref{alg:frame}
that cannot be easily implemented
is finding a new harmonic cycle born in $\har_p(K_{i+1})$
when 
$\fsimp_i$ 
is positive.
We can do this naively by computing a basis for the  kernel of $\delta$ in $K_{i+1}$, and check all   cocycles in the basis to see if they are independent of the existing set of harmonic representatives
maintained for $K_i$. This, however, requires at least $\Theta(m^3)$ per simplex insertion.

For a more efficient implementation,
before  iteration $i$ which processes the insertion 
$K_{i}\inctosp{\fsimp_i} K_{i+1}$,
we maintain the following matrices
for each homology degree $p$:
\begin{enumerate}

    \item $\Harmat^{p}$: Columns in  $\Harmat^{p}$
    are
    partial $p$-representatives 
    for all indices in $U^p$,
 which
 form a \emph{harmonic cycle basis} for $\har_p(K_i)$;
 see \Cref{sec:harmat-basis} for justification.

\item $\Rmat^p$: Columns in $\Rmat^p$ form a \emph{boundary basis} with distinct pivots for $\Bnd_p(K_i)$.
Notice that $\Rmat^p$ is indeed the ``reduced'' matrix 
in the
ordinary persistence algorithm~\cite{cohen2006vines}
with only $p$-chains
and without zero columns.

\item $\Cobmat^p$: Columns in $\Cobmat^p$ form a \emph{coboundary basis}
for $\Bnd^p(K_i)$.


\item 
$\Chnmat^p$: Columns in $\Chnmat^p$
encode $p$-pseudo-cochains in $K_i$ w.r.t\ the basis $\Rmat^p$
whose coboundaries are columns in
$\Cobmat^{p+1}$.
More specifically,
$\Chnmat^p$, $\Cobmat^{p+1}$ have the same number of columns
and for each $j$,
$\Cobmat^{p+1}[j]=\dG(\psi_j)$ 
where $\psi_j$ is the $p$-pseudo-cochain
encoded by $\Chnmat^p[j]$
w.r.t\ 
$\Rmat^p$.

\item $\Bndmat^p$: Columns in $\Bndmat^p$ are boundaries of the
coboundaries in $\Cobmat^{p+1}$,
i.e., 
 $\Bndmat^p$, $\Cobmat^{p+1}$
have the same number of columns and
$\Bndmat^p[j]=\partial(\Cobmat^{p+1}[j])$ for each $j$.
We further ensure that columns in $\Bndmat^p$ have distinct pivots
so that the reduction process in \Cref{algr:reduce} 
can be done in $O(\filtcnt^2)$ time.
\end{enumerate}

In summary,
columns in   $\Chnmat^{p-1}$, $\Cobmat^{p}$,  and  $\Bndmat^{p-1}$
have one-to-one correspondence
as follows:
\[
\Chnmat^{p-1}[j]\xmapsto{\ \dG\ }\Cobmat^{p}[j]\xmapsto{\ \partial\ }\Bndmat^{p-1}[j].\]

Other than $\Harmat^p$ which is already in \Cref{alg:frame},
the rationale for maintaining the matrices is as follows
(assuming $\fsimp_i$ is a $p$-simplex):
\begin{itemize}
    \item We use $\Bndmat^{p-1}$ to find a linear combination $\phi$
    of  $\hat{\fsimp}_i$ and cocycles in $\Cobmat^{p}$ 
    such that $\partial\phi=0$
    (see \Cref{algr:reduce}).
    The cocycle $\phi$ 
    is then a new harmonic $p$-cycle when $\fsimp_i$ is positive.
    \item 
    $\Chnmat^{p-1}$ helps maintain the coboundary basis in $\Cobmat^{p}$.
    \item
     $\Rmat^{p-1}$ provides the boundary basis on which 
    pseudo-cochains in $\Chnmat^{p-1}$ can be defined.
\end{itemize}



\para{Initial reduction.}
When processing the
insertion of 
a $p$-simplex $\fsimp_i$,
we first apply the  reduction  
in \Cref{algr:usual-reduce}
on $\Rmat^{p-1}$
(which is indeed 
performed in the ordinary persistence algorithm~\cite{EdelsbrunnerLetscherZomorodian2002}).
The following
\Cref{prop:usual-reduce} 
is  well-known~\cite{DeyWang2022,EdelsbrunnerLetscherZomorodian2002}:
\begin{proposition}
\label{prop:usual-reduce}
If 
Algorithm~\ref{algr:usual-reduce} ends with $z=0$,
then $\fsimp_i$ is positive;
otherwise, $\fsimp_i$ is negative.
\end{proposition}
Notice that besides determining whether $\fsimp_i$ is positive or negative,
\Cref{algr:usual-reduce}
also ensures that $\Rmat^{p-1}$ represents a basis for $\Bnd_{p-1}(K_{i+1})$
in the next iteration
(details discussed later).

\begin{algorithm}[h]
\caption{Reduction based on $\Rmat^{p-1}$ for determining whether $\fsimp_i$ is positive or negative}

\begin{algorithmic}[1]
\State $z\leftarrow\partial({\sG}_i)$
\While{$z\neq0$ and $\exists j$ s.t.\ $\pivot(\Rmat^{p-1}[j])=\pivot(z)$}
\State $k\leftarrow\pivot(z)$
\State $\alpha_1\leftarrow z[k]$
\State $\alpha_2\leftarrow\Rmat^{p-1}[j][k]$
\State $z\leftarrow z-(\alpha_1/\alpha_2)\cdot\Rmat^{p-1}[j]$
\EndWhile
\end{algorithmic}
\label{algr:usual-reduce}
\end{algorithm}


\para{$\sG_i$ is positive.}
We need to find a new harmonic $p$-cycle born in $K_{i+1}$.
For this, we perform another reduction on $\Bndmat^{p-1}$ and $\Cobmat^{p}$
 in \Cref{algr:reduce}.

\begin{algorithm}[!h]
\caption{Reduction based on $\Bndmat^{p-1},\Cobmat^p$ for finding a new-born harmonic cycle 
}

\begin{algorithmic}[1]
\State $\phi\leftarrow\hat{\sG}_i$
\State $\zG\leftarrow\partial(\hat{\sG}_i)$
\While{$\zG\neq0$ and $\exists j$ s.t.\ $\pivot(\Bndmat^{p-1}[j])=\pivot(\zG)$}
\State $k\leftarrow\pivot(\zG)$
\State $\alpha_1\leftarrow\zG[k]$
\State $\alpha_2\leftarrow\Bndmat^{p-1}[j][k]$
\State $\zG\leftarrow \zG-(\alpha_1/\alpha_2)\cdot\Bndmat^{p-1}[j]$
\State $\phi\leftarrow \phi-(\alpha_1/\alpha_2)\cdot\Cobmat^{p}[j]$
\EndWhile
\end{algorithmic}
\label{algr:reduce}
\end{algorithm}

\begin{remark}
So far we have performed analogous operations on 
$\Rmat^{p-1}$ and $\Bndmat^{p-1}$.
However,
we notice that $\Rmat^{p-1}$ and $\Bndmat^{p-1}$
are in general two different matrices with different purposes.
The main purpose of $\Rmat^{p-1}$ is to form a filtered 
basis for  $({p-1})$-boundary groups
of  complexes in $\filt$.
Columns in $\Rmat^{p-1}$ are {fixed}
once added to $\Rmat^{p-1}$.
On the other hand,
 $\Bndmat^{p-1}$ is mainly for 
finding a new-born harmonic cycle when $\sG_i$ is positive.
Columns in
$\Bndmat^{p-1}$
are not fixed and can be changed in later steps
(such as in \Cref{algr:sum-bndmat}).
\end{remark}

\begin{proposition}
\label{prop:reduce}
If 
Algorithm~\ref{algr:reduce} ends with $\zG=0$,
then $\fsimp_i$ is positive
and $\phi$ is a new harmonic $p$-cycle born in $\har_p(K_{i+1})$;
otherwise, $\fsimp_i$ is negative.
\end{proposition}
\begin{proof}

Since $\hat{\fsimp}_i$ is a new $p$-cocycle in $K_{i+1}$,
we have that $\phi$ is always a new $p$-cocycle in $K_{i+1}$
and  $\zG=\partial(\phi)$ at the end of each iteration in Algorithm~\ref{algr:reduce}.
If Algorithm~\ref{algr:reduce}
ends with $\zG=0$, then 
we have $\partial(\phi)=0$.
This implies that $\phi$ is a new harmonic $p$-cycle born in $\har_p(K_{i+1})$.
By \Cref{thm:harmonic-inclusion},
$\fsimp_i$ is positive.

Now consider the case that \Cref{algr:reduce} ends with $\zG\neq 0$.
For  contradiction, suppose instead that $\fsimp_i$ is positive,
and let $\psi$ be a new harmonic $p$-cycle born in $\har_p(K_{i+1})$.
Since columns in $\Harmat^p$ and $\Cobmat^p$ form a basis of $\Zyc^p(K_i)$
and $\hat{\fsimp}_i$ is a new $p$-cocycle in $K_{i+1}$,
we have that $\Harmat^p$, $\Cobmat^p$, and $\hat{\fsimp}_i$ altogether form a basis of $\Zyc^p(K_{i+1})$.
Hence, 
we have
\mbeg\psi=\aG\cdot\hat{\fsimp}_i+\sum_j\mu_j\cdot\Harmat^p[j]+\sum_k\lG_k\cdot\Cobmat^p[k],\mend
where $\aG\neq 0$ because 
$\psi$ is a new  $p$-cocycle in $K_{i+1}$.
So we have \begin{align*}
0=\partial(\psi)&=\aG\cdot\partial(\hat{\fsimp}_i)+\sum_j\mu_j\cdot\partial(\Harmat^p[j])+\sum_k\lG_k\cdot\partial(\Cobmat^p[k])\\
&=\aG\cdot\partial(\hat{\fsimp}_i)+0+\sum_k\lG_k\cdot\Bndmat^{p-1}[k].
\end{align*}
Now \mbeg\partial(\hat{\fsimp}_i)=\sum_k(-\lG_k/\aG)\cdot\Bndmat^{p-1}[j],\mend
i.e., $\partial(\hat{\fsimp}_i)$ is a linear combination of
columns of $\Bndmat^{p-1}$.
Since the columns of $\Bndmat^{p-1}$ have distinct pivots,
the reduction in Algorithm~\ref{algr:reduce} must end with $\zG=0$,
which is a contradiction.
\end{proof}

By \Cref{prop:reduce},
\Cref{algr:reduce}
must end with $\zG=0$
and
we add $\phi$ as new column to $\Harmat^p$.
Now
consider 
the pseudo-cochain \mbeg\psi_j:\Bnd_{p-1}(K_i)\to\Real\mend in $K_i$
that a column $\Chnmat^{p-1}[j]$ encodes.
Since $\Bnd_{p-1}(K_i)=\Bnd_{p-1}(K_{i+1})$~\cite{DeyWang2022,EdelsbrunnerLetscherZomorodian2002},
we can directly treat $\psi_j$ as a pseudo-cochain   in $K_{i+1}$.
However, 
while we still have 
\mbeg\dG(\psi_j)(\fsimp)=\Cobmat^p[j](\fsimp)\text{ for each $p$-simplex }\fsimp\in K_i,\mend
we may have 
\mbeg\dG(\psi_j)(\fsimp_i)=\psi_j(\partial\fsimp_i)\neq 0,\mend
so that \mbeg\Cobmat^p[j]\neq\dG(\Chnmat^{p-1}[j])\mend in $K_{i+1}$.
To deal with this,
we first allocate two new matrices 
$\bar{\Cobmat}^p$,
$\bar{\Bndmat}^{p-1}$
and let
\mbeg\bar{\Cobmat}^p[j]=\dG(\Chnmat^{p-1}[j]):=\dG(\psi_j)\text{, }\bar{\Bndmat}^{p-1}[j]=\partial(\bar{\Cobmat}^p[j])
\text{ for each $j$,}\mend
where 
all the chains and (pseudo-)cochains
are now defined in $K_{i+1}$.
To get $\bar{\Cobmat}^p[j]$,
we only need to know the value of $\psi_j(\partial\fsimp_i)$.
From the reduction in \Cref{algr:usual-reduce},
we know how to express $\partial\fsimp_i$ as a formal sum of columns in $\Rmat^{p-1}$,
i.e., \mbeg\partial\fsimp_i=\sum_{k}\mu_k\Rmat^{p-1}[k],\mend
where each non-zero $\mu_k$ is given by the ``$\aG_1/\aG_2$'' in \Cref{algr:usual-reduce}.
Then,
\mbeg\psi_j(\partial\fsimp_i)=\psi_j(\sum_{k}\mu_k\Rmat^{p-1}[k])
=\sum_{k}\mu_k\psi_j(\Rmat^{p-1}[k])
=\sum_{k}\mu_k\Chnmat[j][k].\mend

Due to the previous update, columns in $\bar{\Bndmat}^{p-1}$ may not have distinct pivots.
We then 
run \Cref{algr:sum-bndmat}
to make columns in $\bar{\Bndmat}^{p-1}$ have distinct pivots again.
To see how \Cref{algr:sum-bndmat} works,
notice that before executing the for loop in line~\ref{line:sum-bndmat-for}, we have
\mbeg\bar{\Cobmat}^p[\lG_j]={\Cobmat}^p[\lG_j]+\aG_j\hat{\fsimp}_i
\text{ for }j\in\{1,\ldots,\ell\}.\mend
This means that
\mbeg\bar{\Bndmat}^{p-1}[\lG_j]={\Bndmat}^{p-1}[\lG_j]+\aG_j\partial(\hat{\fsimp}_i).\mend
So for $j>1$, we have
\begin{align*}
&\bar{\Bndmat}^{p-1}[\lG_j]-(\aG_j/\aG_1)\cdot\bar{\Bndmat}^{p-1}[\lG_1]\\
=&{\Bndmat}^{p-1}[\lG_j]+\aG_j\partial(\hat{\fsimp}_i)
-(\aG_j/\aG_1)({\Bndmat}^{p-1}[\lG_1]+\aG_1\partial(\hat{\fsimp}_i))\\
=&{\Bndmat}^{p-1}[\lG_j]-(\aG_j/\aG_1)\cdot{\Bndmat}^{p-1}[\lG_1].
\end{align*}
Since $\pivot({\Bndmat}^{p-1}[\lG_1])<\pivot({\Bndmat}^{p-1}[\lG_j])$,
we have \mbeg\pivot(\bar{\Bndmat}^{p-1}[\lG_j])=\pivot({\Bndmat}^{p-1}[\lG_j])\mend
after executing the for loop.
Now we have \mbeg\bar{\Bndmat}^{p-1}[\lG]={\Bndmat}^{p-1}[\lG]\text{ for }\lG\not\in\{\lG_1,\ldots,\lG_\ell\},\mend
meaning that
before executing the while loop
in line~\ref{line:sum-bndmat-while},
only pivots of the two columns indexed by $\lG_1$ may differ in
$\bar{\Bndmat}^{p-1}$ and ${\Bndmat}^{p-1}$.
This means that  at most two columns in
$\bar{\Bndmat}^{p-1}$ may have the same pivot
before executing the while loop.
We then have that,
in each iteration of the while loop,
at most two columns in
$\bar{\Bndmat}^{p-1}$ have the same pivot.
Hence, columns in $\bar{\Bndmat}^{p-1}$
have distinct pivots after finishing the loop.
Notice that the pivot $k$ in line~\ref{line:sum-bndmat-k}
decreases in each iteration,
meaning that  the while loop executes for no more than $i$ iterations.
Hence, \Cref{algr:sum-bndmat} runs in $O(\filtcnt^2)$ time.
Notice that allocating the two new matrices
$\bar{\Cobmat}^p$,
 $\bar{\Bndmat}^{p-1}$
are purely for clearly presenting and justifying the operations
in \Cref{algr:sum-bndmat}.
To implement this in computer, 
the operations can be directly performed  on 
${\Cobmat}^p$
and ${\Bndmat}^{p-1}$.

\begin{algorithm}[h]
\caption{Making columns in $\bar{\Bndmat}^{p-1}$ have distinct pivots}

\begin{algorithmic}[1]
\State let $\bar{\Cobmat}^p[\lG_1],\ldots,\bar{\Cobmat}^p[\lG_\ell]$ be all the columns in
$\bar{\Cobmat}^p$ s.t.\ $\bar{\Cobmat}^p[\lG_j](\fsimp_i)\neq 0$
for each $j$
\State reorder $\lG_1,\ldots,\lG_\ell$ s.t.\ $\pivot({\Bndmat}^{p-1}[\lG_1])<\pivot({\Bndmat}^{p-1}[\lG_j])$
for each $j>1$
\State $\aG_1\leftarrow\bar{\Cobmat}^p[\lG_1](\fsimp_i)$
\For{$j=2,\ldots,\ell$}\label{line:sum-bndmat-for}
\State $\aG_j\leftarrow\bar{\Cobmat}^p[\lG_j](\fsimp_i)$
\State ${\Chnmat}^{p-1}[\lG_j]\leftarrow{\Chnmat}^{p-1}[\lG_j]-(\aG_j/\aG_1)\cdot{\Chnmat}^{p-1}[\lG_1]$
\State $\bar{\Cobmat}^p[\lG_j]\leftarrow\bar{\Cobmat}^p[\lG_j]-(\aG_j/\aG_1)\cdot\bar{\Cobmat}^p[\lG_1]$
\State $\bar{\Bndmat}^{p-1}[\lG_j]\leftarrow\bar{\Bndmat}^{p-1}[\lG_j]-(\aG_j/\aG_1)\cdot\bar{\Bndmat}^{p-1}[\lG_1]$
\EndFor
\While{exist two columns $\bar{\Bndmat}^{p-1}[\lG]$, $\bar{\Bndmat}^{p-1}[\mG]$ with same pivot}\label{line:sum-bndmat-while}
\State $k\leftarrow\pivot(\Bndmat^{p-1}[\lG])$\label{line:sum-bndmat-k}
\State $\gamma_1\leftarrow\Bndmat^{p-1}[\lG][k]$
\State $\gamma_2\leftarrow\Bndmat^{p-1}[\mG][k]$
\State ${\Chnmat}^{p-1}[\lG]\leftarrow{\Chnmat}^{p-1}[\lG]-(\gamma_1/\gamma_2)\cdot{\Chnmat}^{p-1}[\mG]$
\State $\bar{\Cobmat}^p[\lG]\leftarrow\bar{\Cobmat}^p[\lG]-(\gamma_1/\gamma_2)\cdot\bar{\Cobmat}^p[\mG]$
\State $\bar{\Bndmat}^{p-1}[\lG]\leftarrow\bar{\Bndmat}^{p-1}[\lG]-(\gamma_1/\gamma_2)\cdot\bar{\Bndmat}^{p-1}[\mG]$
\EndWhile
\State{${\Cobmat}^p\leftarrow\bar{\Cobmat}^p$}
\State{${\Bndmat}^{p-1}\leftarrow\bar{\Bndmat}^{p-1}$}
\end{algorithmic}
\label{algr:sum-bndmat}
\end{algorithm}

Finally, we have that
the other matrices do not need to be changed
by \Cref{prop:dim-chng}.



\para{$\sG_i$ is negative.}
First update $\Harmat^{p-1}$ 
according to \Cref{alg:frame}.
Since $\fsimp_i$ is negative, 
we have $z\neq 0$ at the end of \Cref{algr:usual-reduce}.
We then add $z$ to $\Rmat^{p-1}$ 
to form a basis for $\Bnd_{p-1}(K_{i+1})$
(see~\cite{DeyWang2022,EdelsbrunnerLetscherZomorodian2002}).
For ease of presentation, 
let $\Rmat^{p-1}$ still denote 
the matrix
before $z$ is added
and let $\bar{\Rmat}^{p-1}=\Rmat^{p-1}\cup\{z\}$.
We extend each column 
of $\Chnmat^{p-1}$
by appending 0
so that columns of $\Chnmat^{p-1}$
encode $(p-1)$-pseudo-cochains in $K_{i+1}$
w.r.t\ $\bar{\Rmat}^{p-1}$.
Notice that we may also have
$\dG(\Chnmat^{p-1}[j])(\fsimp_i)\neq 0$
for a $j$.
Hence,
we update 
$\Cobmat^p$ and $\Bndmat^{p-1}$
as in the previous case where $\fsimp_i$ is positive,
including performing \Cref{algr:sum-bndmat}.
After this, 
${\Cobmat}^p[j]=\dG(\Chnmat^{p-1}[j])$, ${\Bndmat}^{p-1}[j]=\partial({\Cobmat}^p[j])$
in $K_{i+1}$
for each $j$
and columns in ${\Bndmat}^{p-1}$ have distinct pivots.

Since $\dim(\Bnd^{p}(K_{i+1}))=\dim(\Bnd^{p}(K_{i}))+1$
by \Cref{prop:dim-chng},
we need to add a new column to $\Cobmat^p$ 
to form a basis for $\Bnd^{p}(K_{i+1})$.
For this,
we first add the column 
\mbeg\vec{v}=(0,\ldots,0,1)^T\mend
to $\Chnmat^{p-1}$
and observe the following:
\begin{proposition}\label{prop:sG-i-new-cob}
Let $\vec{v}$ be as defined above.
Then,
$\hat{\fsimp}_i=\dG(\vec{v})$.
\end{proposition}
\begin{proof}
Let $r$ be the length of $\vec{v}$
and let $\psi$ be the pseudo-cochain encoded by $\vec{v}$
w.r.t 
$\bar{\Rmat}^{p-1}$.
For a $p$-simplex $\fsimp\in K_i$,
we have 
\mbeg
\dG(\vec{v})(\fsimp)=\psi(\partial\fsimp)=\psi(\sum_{k=1}^{r-1}\mu_k{\Rmat}^{p-1}[k])
=\sum_{k=1}^{r-1}\mu_k\psi(\Rmat^{p-1}[k])
=0.
\mend
Also, by the reduction in \Cref{algr:usual-reduce},
we have \mbeg\partial\fsimp_i=\sum_{k=1}^{r-1}\lG_k{\Rmat}^{p-1}[k]+z.\mend
Then,
\mbeg\dG(\vec{v})(\fsimp_i)=\psi(\partial\fsimp_i)=\psi(\sum_{k=1}^{r-1}\lG_k{\Rmat}^{p-1}[k]+z)
=\psi(z)=1.\qedhere\mend
\end{proof}

\begin{proposition}\label{prop:add-to-cob-mat}
The columns in ${\Cobmat}^p$ 
along with 
$\hat{\fsimp}_i$ 
form a basis for $\Bnd^p(K_{i+1})$.
\end{proposition}
\begin{proof}
Abuse the notations slightly by
letting
${\Cobmat}^p$
and $\bar{\Cobmat}^p$ denote  versions of the two matrices before executing \Cref{algr:sum-bndmat},
i.e., columns in ${\Cobmat}^p$ are still cochains in $K_i$ 
forming a basis of $\Bnd^p(K_i)$
and
each $\bar{\Cobmat}^p[j]$ (a cochain in $K_{i+1}$)
is simply derived by taking the coboundary 
of $\Chnmat^{p-1}[j]$ in  $K_{i+1}$.
We first show that columns of
$\bar{\Cobmat}^p[j]$ and $\hat{\fsimp}_i$  are linearly 
independent.
Suppose that \mbeg\sum_j\aG_j\bar{\Cobmat}^p[j]+\gamma\hat{\fsimp}_i=0.\mend
Let $\iG^\sharp:\Chn^p(K_{i+1})\to\Chn^p(K_{i})$
be the cochain map induced by inclusion.
We have that \mbeg\sum_j\aG_j{\Cobmat}^p[j]=\sum_j\aG_j\iG^\sharp(\bar{\Cobmat}^p[j])
=\iG^\sharp(\sum_j\aG_j\bar{\Cobmat}^p[j])
=\iG^\sharp(-\gamma\hat{\fsimp}_i)
=-\gamma\iG^\sharp(\hat{\fsimp}_i)=0,\mend
meaning that each $\aG_j$ is zero.
This also means that $\gamma=0$  because $\hat{\fsimp}_i\neq 0$.
We then have that
columns of
$\bar{\Cobmat}^p[j]$ along with $\hat{\fsimp}_i$
form a basis of $\Bnd^{p}(K_{i+1})$
because $\hat{\fsimp}_i\in\Bnd^{p}(K_{i+1})$
(\Cref{prop:sG-i-new-cob}) and $\dim(\Bnd^{p}(K_{i+1}))=\dim(\Bnd^{p}(K_{i}))+1$ (\Cref{prop:dim-chng}).
\Cref{prop:add-to-cob-mat}
then follows because   \Cref{algr:sum-bndmat}
does not change the linear span of matrix columns.
\end{proof}

By \Cref{prop:sG-i-new-cob,prop:add-to-cob-mat},
we add $\hat{\fsimp}_i$
to $\Cobmat^p$
and add $\partial(\hat{\fsimp}_i)$
to $\Bndmat^{p-1}$.
Notice that there may be at most two columns having the same pivot
in $\Bndmat^{p-1}$ after this.
For this, we apply a procedure similar to the while loop in 
\Cref{algr:sum-bndmat}
to make $\Bndmat^{p-1}$ have distinct pivots again.

By \Cref{prop:dim-chng},
the other matrices do not need to be changed.

\bigskip

We now  conclude the following:
\begin{theorem}
    The \hcb{} of $\filt$ can be computed in $O(\filtcnt^3)$ time.
\end{theorem}
\begin{proof}
Each iteration of the algorithm takes no more than $O(\filtcnt^2)$ time
and so
the overall complexity 
is $O(\filtcnt^3)$.
Correctness of our implementation 
follows from 
correctness of the operations detailed in this section.
\end{proof}

\section{Sublevel Set Harmonic Chain Barcode and Its Stability}
\label{sec:stability}

In this section, we introduce the notion of \emph{sublevel set harmonic chain barcodes} and present our stability results based on this notion.
Our stability proof makes use of the work on the algebraic stability of block decomposable $\Rspace^2$-modules~\cite{bjerkevik2021stability,botnan2018algebraic}. 
In brief, we ``lift'' the 1-parameter harmonic zigzag module to an $\Rspace^2$-indexed 2-parameter persistence module which is block decomposable. 
We then show that in the typical setting of sublevel set filtrations, there exists an interleaving between the lifted modules. Our main work here is to construct a concrete extension (see~\cref{subsec:extension}) and an interleaving between extensions to $\Rspace^2$-modules (see~\cref{subsec:stability}) realizing the category-theoretic concepts employed in~\cite{bjerkevik2021stability,botnan2018algebraic}. From the Isometry Theorem~\cite{bauer2014induced,botnan2018algebraic}, we then deduce the stability of the sublevel set {\hcb}.
We also address the disparity between the natural {\hcb} of a sublevel set filtration (which consists of closed-open intervals\footnote{The closed-open intervals over the real values we have here are different from the closed-open intervals of~\cite{bjerkevik2021stability,botnan2018algebraic} over the integers. Indeed, when considering only integer indices, we only have closed-closed intervals in our setting. The extension of our closed-open interval would be similar to a closed block of~\cite{botnan2018algebraic} but with the right vertical side open.}) and common conventions of zigzag persistence (which work with closed-closed intervals in our setting).  

\subsection{Interleaving and Stability for Ordinary Persistence} \label{subsec:ordinaryinterleaving}

Let $F$ and $G$ be two filtrations over the complexes $K$ and $K'$ respectively, and let the maps connecting the complexes in $F$ and $G$ be $f^{s,t}$ and $g^{s,t}$ respectively, for  $s\leq t \in \Rspace$. Let $C(F)$ and $C(G)$ be the corresponding filtrations of chain groups, and $H(F)$ and $H(G)$ the corresponding persistence modules. With an abuse of notation, we denote the maps induced on chain groups and homology groups also by $f^{s,t}$ and $g^{s,t}$ respectively, where $s\leq t \in \Rspace$.

\begin{definition}
\label{def:chain-interleaving}
Let $F$ and $G$ be two filtrations over  $K$ and $K'$ respectively. 
An \emph{$\eps$-chain-interleaving between $F$ and $G$} (or an $\eps$-interleaving at the chain level) is given by two sets of homomorphisms $\{\phi_\alpha : C(K_\alpha) \xrightarrow{} C(K'_{\alpha+\eps}) \}$ and $\{\psi_\alpha : C(K'_{\alpha}) \xrightarrow{} C(K_{\alpha+\eps}) \}$, such that 
\begin{enumerate}
    \item $\{ \phi_\alpha \}$ and $\{\psi_\alpha \}$ commute with the maps in the filtration, that is, for all $\alpha,t \in \Rspace$, $g^{\alpha+\eps,\alpha+\eps + t} \phi_\alpha = \phi_{\alpha+t} f^{\alpha, \alpha+t}$ and $f^{\alpha+\eps,\alpha+\eps+t} \psi_\alpha = \psi_{\alpha+t} g^{\alpha, \alpha+t}$; 
    \item The following diagrams commute: 
    \begin{equation}
      \begin{tikzcd}
  C(K_\alpha)
    \ar[r,"f^{\alpha, \alpha+\eps}"]
    \ar[dr, "\phi_\alpha", very near start, outer sep = -2pt]
  & C(K_{\alpha+\eps})
    \ar[r, "f^{\alpha+\eps, \alpha+2\eps}"]
    \ar[dr, "\phi_{\alpha+\eps}", very near start, outer sep = -2pt]
  & C(K_{\alpha+2\eps})
  \\
  C(K'_\alpha)
    \ar[r, swap,"g^{\alpha, \alpha+\eps}"]
    \ar[ur, crossing over, "\psi_\alpha"', very near start, outer sep = -2pt]
  & C(K'_{\alpha+\eps})
    \ar[r, swap,"g^{\alpha+\eps, \alpha+2\eps}"]
    \ar[ur, crossing over, "\psi_{\alpha+\eps}"', very near start, outer sep = -2pt]
  & C(K'_{\alpha+2\eps}).
  \end{tikzcd}
    \end{equation}
\end{enumerate}
\noindent The \emph{chain interleaving distance} is defined to be 
\begin{equation}
d_{CI}(F, G) := \inf \{\eps\geq 0 \mid \text{there exists an $\eps$-chain interleaving between } F\; \text{and}\; G\}.
\end{equation}
\end{definition}

We notice that the standard notion of $\eps$-interleaving~\cite{ChazalCohenSteinerGlisse2009}, denoted $d_I(F,G)$, is analogous to our definition above. However, it is defined on the persistence modules $H(F)$ and $H(G)$ on the homology level~\cite{ChazalCohenSteinerGlisse2009}, rather than the filtration of chain groups.  
Filtrations on the chain level arising from sublevel-sets of a function are ones that interest us in this paper.


We refer to~\cite{ChazalCohenSteinerGlisse2009} for proof of the following 
stability theorem for ordinary persistence.

\begin{theorem}
\label{theorem:persistence-stability}
    Let $F$ and $G$ be filtrations defined over (finite) complexes $K$ and $K'$, respectively. Then 
    $$ d_B( \Dgm(F), \Dgm(G)) \leq d_I(F,G).$$
\end{theorem}

See~\cite{Cohen-SteinerEdelsbrunnerHarer2005, ChazalCohenSteinerGlisse2009} for more on the stability of ordinary persistence. We also notice that in the above theorem we can replace $d_{CI}$ in place of $d_I$, since the existence of an $\eps$-chain-interleaving implies the existence of an $\eps$-interleaving on the homology level. 

\subsection{Sub-level-Set Filtration}\label{sec:Sub-level-Set-Filtration}

Note that we assumed up until now a filtration to be simplex-wise. While this is not true for a sublevel set filtration, we shall construct simplex-wise filtrations from a general filtration. 

\begin{definition}\label{dfn:bd-span-bar2}
Let $\tilde{F}$ be an increasing filtration of complex $K$ which is not necessarily simplex-wise and let $z \in Z_p(K)$ be a $p$-cycle. 
Let $t_{b(z)}$ be the time when $z$ is born,  i.e., $z\in Z_p(K_{t_{b(z)}})\setminus Z_p(K_{t_{b(z)-1}})$, and $t_{d(z)}$ be the time when $z$ dies as a harmonic chain, i.e., $\delta(z)\neq 0 $ in $K_{t_{d(z)}}$, but $\delta(z)=0$ in $K_t$ for $t_{b(z)} \leq t < t_{d(z)}$. Define the \emph{span} of $z$ as $span(z)=[t_{b(z)}, t_{d(z)})$.  If $z$ is not harmonic at birth, its span is the empty set (notice that a cycle  continues to be non-harmonic once it becomes so).
    The \emph{bar} of $z$ is defined as $bar(z) := [t_{b(z)}, t_{d(z)-1}]$.
\end{definition}

Let $\tilde{F}$ be an increasing filtration of $K$ which is not necessarily simplex-wise, i.e., two consecutive complexes in $\tilde{F}$ can differ by more than one simplex. We fix an ordering, once and for all, for vertices of $K$, which also fixes an ordering for all $K$'s simplices (using, say, the lexicographic ordering of the ordered sequence of vertices in a simplex). Given $\delta>0$ small enough, 
we expand $\tilde{F}$ into a simplex-wise filtration $\tilde{F}(\delta)$
by  expanding each inclusion $K_{t_{i-1}}\incto K_{t_{i}}$
in $\tilde{F}$
starting with $i=1$.
To differentiate,
 denote each complex in $\tilde{F}(\delta)$
as $K'_t$ where $t\in\mathbb{R}$ is a timestamp.
In each iteration, we expand the following inclusion in $\tilde{F}$:
$$K_{t_{i-1}} \incto K_{t_i} 
$$ 
 into  several simplex-wise inclusions in $\tilde{F}(\delta)$: 
$$  K_{t_{i-1}}=K'_{s} \incto K'_{t_i-\delta} \incto K'_{t_i}\incto K'_{t_i +\delta} \incto \cdots \incto K'_{t_{i+(k-1)\delta}} =K_{t_i}.
$$ 
In the above, $K'_{t_i-\delta}$ is a ``dummy'' complex equal to $K'_{s}$
(so $K'_{s} \incto K'_{t_i-\delta}$ is an identity map)
and $k$ is the number of simplices added from $K_{t_{i-1}}$ to $K_{t_i}$ in $\tilde{F}$.
In short, adding the dummy complexes makes sure that the closed bars of the simplex-wise filtration 
``induce''
the closed-open time intervals in $\tilde{F}$ in which a cycle is harmonic; see \Cref{sec:stability} for details.
We also have that the $k$ number of simplices are added based on
their simplex ordering described above.

Let $span_{\delta}(z)$ denote the span of $z$ in $\tilde{F}(\delta)$, and analogously define $bar_{\delta}(z)$, $b_{\delta}(z)$, and $d_{\delta}(z)$. We observe that addition of dummy complexes implies that for a cycle $z$, bar of $z$ in $\tilde{F}(\delta)$ must approximate the span of $z$ in $\tilde{F}$; see Fig.~\ref{fig:spanbar} for an example. In more detail we have the following.

\begin{figure}
    \centering
    \includegraphics[width=0.7\linewidth]{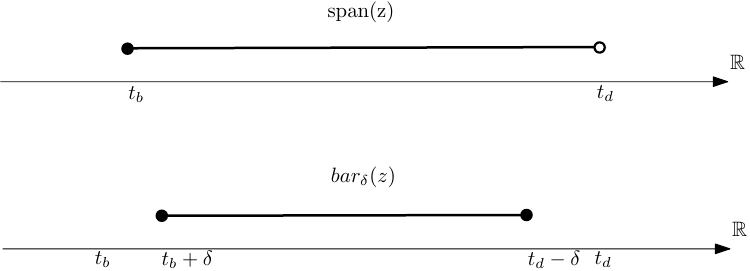}
    \caption{Turning closed-open spans into closed bars}
    \label{fig:spanbar}
\end{figure}

\begin{lemma}\label{lemma:barsnapping}
Let $m$ be the size of $K$
and let $z\in Z_*(K)$ be a cycle. For $\delta>0$ sufficiently small,
$t_{b(z)} \leq t_{b_\delta(z)} \leq t_{b(z)} +m\delta $, and $t_{d(z)} \leq t_{d_\delta(z)} \leq t_{d(z)} +m\delta$. Moreover, we have $ t_{d(z)} -\delta \leq t_{d_\delta(z)-1} \leq t_{d(z)}+(m-1)\delta$.

\end{lemma}

\begin{proof}

Since $K_{t_i-\delta}= K_{t_{i-1}}$, we have $b_\delta (z) \geq b(z)$.  It is clear that $b_\delta(z) \leq b(z)+m\delta$. 
The inequalities for $d_\delta(z)$ can also be easily checked. 

$ t_{d(z) -1} $ is the largest (real valued) index smaller than $t_d(z)$. The smallest that this index can be is the index of a dummy complex added before $K_{t_{d(z)}}$, and the largest it can be is the one to last index among complexes added after $K_{d(z)}$. 
\end{proof}

Therefore, the barcode of $\tilde{F}(\delta)$ can induce
a barcode for $\tilde{F}$.  Indeed, for any $\delta$ small enough, there is an obvious mapping $\Gamma: \tilde{F}(\delta) \xrightarrow{} \tilde{F}$. 
\begin{definition}
Let $F$ be a simplex-wise filtration (as in \Cref{sec:persistence}) and
let $\obarc(F)$ be the \emph{closed-open barcode of $\har{(F)}$}, i.e., $[t_b,t_{d}) \in \obarc(F)$ iff $[t_b,t_{d-1}] \in \hbarc(F)$.  
 Moreover, let $\tilde{F}$ and $\tilde{F}(\delta)$ be as defined in this section.
We define the \emph{harmonic chain barcode of $\tilde{F}$} as the set $\obarc(\tilde{F}) := \Gamma(\obarc(\tilde{F}(\delta)))$ for $\delta$ small enough.
\end{definition}

Note that $\obarc(\tilde{F})$ is the limit of $\hbarc(\tilde{F}(\delta))$ as $\delta$ approaches 0. 

In the rest of this section, we prove stability with respect to simplex-wise filtrations and at the end we show the stability for sublevel set harmonic barcode.

\subsection{$\mathbb{U}$-indexed Persistence Modules}\label{subsec:extension}

Let $\Rspace^{op}$ denote the poset $\Rspace$ with opposite ordering. Let $\mathbb{U}$ be the subset of $\Rspace^{op}\times \Rspace$ consisting of pairs $(a,b), a\leq b$. Note that $(a,b) < (c,d)$ if and only if $c<a, b < d$. Therefore, $\mathbb{U}$ can be considered as consisting of real-valued intervals with the order being inclusions between the intervals. A \emph{$\mathbb{U}$-indexed persistence module} $M$ is a functor from the poset $\mathbb{U}$ to the category of vector spaces. We have for each $(a,b)$, a vector space $M_{(a,b)}$, and for each relation $(a,b) < (c,d)$, a linear map\footnote{For simplicity we omit the subscripts for maps as long as this does not cause confusions.} $f:M_{(a,b)} \xrightarrow{} M_{(c,d)}$.

We now define the \emph{lift} (or \emph{extension}) of the harmonic zigzag module $\har(\filt)$, denoted $E= E(\har(\filt))$. We have $$ E(a,b) = \{ z \in Z_*(K)\mid b(z) \leq b ,  d(z)  \geq a) \}.$$ That is, $E(a,b)$ contains the chains that are born at or before $b$ and die at or after $a$ (see \Cref{dfn:bd-span-bar2}). Observe that if $z$ and $z'$ are born at or before $t$, then $z+z'$ is born at or before $t$, and if $z$ and $z'$ die at or after $t$, then $z+z'$ dies at or after $t$. To see the last statement, let $d$ be the minimum  time that $z$ and $z'$ die, and assume $z_1+z_2$ dies at $t_i < d$. Then $z+z' \notin K_{t_{i+1}} \subset K_{t_i}$. Therefore, one of $z$ and $z'$ must also not be in the image, which is a contradiction. From these observations we deduce the following.

\begin{lemma}
    For any $(a,b) \in \mathbb{U}$, $E(a,b)$ is a vector space. 
\end{lemma}

\begin{lemma}\label{l:directsum}
    The extension operator $E$ respects direct sums. 
\end{lemma}

If $(a,b) < (c,d) $, then $[a,b] \subset [c,d]$ and we have an inclusion $E(a,b) \subset E(c,d)$. These are maps of the extension $\mathbb{U}$-module. Note that \textit{each} cycle of $K$ is in some $E(a,b)$ 
where $a$ could be possibly equal to $b$.

\subsection{Blocks and Block-Decomposable Modules}

An \emph{interval} in a $\mathbb{U}$-Module $M$ is a subset ${\cal{J}} \subset \mathbb{U}$ such that 
\begin{enumerate}
    \item ${\cal{J}}$ is non-empty,
    \item if $p,q \in {\cal{J}}$ and $p \leq r \leq q$, then $r \in {\cal{J}},$ and,
    \item for any $p,q \in {\cal{J}}$, there is a sequence $p = r_0, r_1, \ldots, r_l = q$ of elements of ${\cal{J}}$ such that for all $0 \leq j \leq l-1$, $r_j$ and $r_{j+1}$ are comparable (connectivity). 
\end{enumerate}

A multiset of intervals in $\mathbb{U}$ is called a \emph{barcode}.

For ${\cal{J}}$ an interval in $\mathbb{U}$, the \emph{interval module $I^{{\cal{J}}}$} is a $\mathbb{U}$-indexed module such that $I^{{\cal{J}}}_p = \Rspace$ if $p \in {\cal{J}}$, and $I_a^{{\cal{J}}}=0$ otherwise. The maps satisfy $f:p\xrightarrow{} q = id_{\Rspace}$ if $p \leq q \in {\cal{J}}$ and $0$ otherwise.

The intervals for a $\mathbb{U}$-Module are indecomposable modules analogous to the interval modules in an $\Rspace$-indexed persistence module~\cite[Prop. 2.2]{botnan2018algebraic}. However, contrary to the 1D persistence, a general $\mathbb{U}$-indexed module cannot be  decomposed into interval modules. Note that the harmonic zigzag module is decompasable into interval modules. This implies that the module $E$ is decomposable into special forms of intervals called blocks.

For a closed interval $[a,b] \subset \Rspace^+$ the \emph{closed block} ${\mathcal{J}_{BL}} = [a,b]_{BL}$ is an interval in $\mathbb{U}$ defined by $$ [a,b]_{BL} = \{ (x,y) \in \mathbb{U}\mid  x \leq b,  y \geq a \}.$$ See~\cref{fig:block-extension}. 
A multiset of blocks in $\mathbb{U}$ is called a \emph{block barcode}.

For any block ${\mathcal{J}_{BL}}$, the interval module $I^{{\mathcal{J}_{BL}}}$ is called a \emph{block module}. 
A $\mathbb{U}$-indexed module is \emph{block decomposable} if it can be written as a direct sum of block modules. 

\begin{figure}
    \centering
    \includegraphics[width=0.5\linewidth]{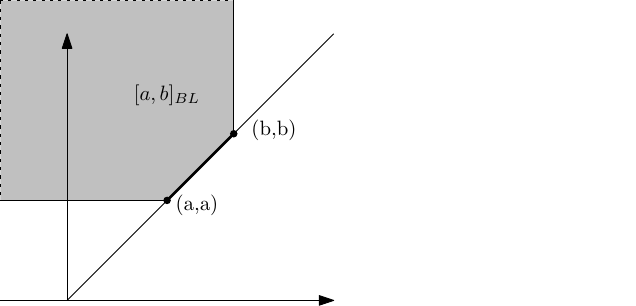}
    \caption{The extension of a closed interval is a closed block.}
    \label{fig:block-extension}
\end{figure}

It is not difficult to check that if $[a,b]$ is a closed interval for a zigzag module, then $E(a,b)$ is a closed block module. Since the direct sum is preserved by $E$, we have  the following.

\begin{lemma}
    The $\mathbb{U}$-indexed module $E(\har(\filt))$ is block decomposable. 
\end{lemma}

\subsection{Interleaving of $\mathbb{U}$-Indexed Modules}

To define $\eps$-interleaving between $\mathbb{U}$-indexed modules, we need to first define a \emph{shift operator} $S^\eps$. For a $\mathbb{U}$-indexed module $M$, we have a shifted module $S^\eps(M)$ such that 
$S^\eps(M)_{(a,b)} =M_{(a-\eps, b+\eps)}$, and for $f:M_{(c,d)} \xrightarrow{} M_{(a,b)}$, $S^\eps(f) = f: M_{(c-\eps,d+\eps)}\xrightarrow{} M_{(a-\eps,b+\eps) )}$.


Given $\eps \geq 0$, two $\mathbb{U}$-indexed modules $M ,N$ are \emph{$\eps$-interleaved} if there exist a pair of morphisms $\Phi: M \xrightarrow{} S^\eps (N)$ and $\Psi: N \xrightarrow{} S^\eps (M)$ such that for each $(a,b) \in \mathbb{U}$

$$ \Psi_{(a-\eps, b+\eps)}  \circ \Phi_{(a,b)} = f:M_{(a,b)} \xrightarrow{} M_{(a-2\eps, b+2\eps)}  $$
and
$$ \Phi_{(a-\eps, b+\eps)}  \circ \Psi_{(a,b)} = g:N_{(a,b)} \xrightarrow{} N_{(a-2\eps, b+2\eps)},$$
where  $g$ in the above denotes  maps of $N$.

The \emph{interleaving-distance} between two $\mathbb{U}$-indexed modules $M$ and $N$ is defined as $$ d_I(M,N) = \inf \{\eps\mid \text{there is an $\eps$-interleaving between $M$ and $N$} \}.$$



\subsection{Stability of Harmonic Chain Barcode}\label{subsec:stability}

Let $F$ be a filtration. Define $\rho(F) = \max \{ t_{d(z)} - t_{d(z)-1} \mid z\in Z_*(K)\}$, that is, $\rho(F)$ the maximum distortion caused by turning spans into closed bars.

\begin{theorem}

Let $F$ and $G$ be simplex-wise increasing filtrations where the maps of $F$ and $G$ are inclusions. Assume that the interleaving distance at chain level between $F$ and $G$ is obtained via inclusions. 
Then $$ d_B( {\hbarc(F)}, {\hbarc(G)} ) \leq d_{I}(F,G) +\max \{\rho(F), \rho(G) \}.$$
\end{theorem}

\begin{proof}
    
    Set $M=E(\har(F))$ and $N=E(\har(G))$.  Let $\eps = d_{I}(F,G)$. We define $\Phi_{(a,b)}: M_{(a,b)} \xrightarrow{} N_{(a-\eps-\rho(F), b+\eps+\rho(F))}$ to be the inclusion. Note that $\Phi$ is well-defined since all the cycles exist in the extension. 
    The morphism $\Psi = \{ \Psi_{(a,b)}\mid (a,b)\in \mathbb{U} \}$ is defined analogously. Since all the maps are inclusions, the morphisms $\Phi$ and $\Psi$ define an $\eps$-interleaving of $M$ and $N$. It follows that $d_I (M,N) \leq \eps+\max \{\rho(F), \rho(G) \}$.
    From the isometry theorem for block decomposable $\mathbb{U}$-indexed modules~\cite{bjerkevik2021stability,botnan2018algebraic}, we deduce that $d_B( {\mathcal{D}} (M), {\mathcal{D}} (N)) = d_I(M,N) \leq \eps +\max \{\rho(F), \rho(G) \}$, where ${\mathcal{D}}(M)$ and ${\mathcal{D}}(N)$ are the block barcodes of $M$ and $N$ respectively. Any matching of block barcodes induces a matching of the parts of the blocks on the diagonal which form the barcodes of the 1D filtrations $\har(F)$ and $\har(G)$ and vice versa (\cite[Proposition 4.2]{botnan2018algebraic}). It follows that $d_B (\hbarc(F), \hbarc(G)) = d_B( {\mathcal{D}} (M), {\mathcal{D}} (N))  \leq \eps + \max \{\rho(F), \rho(G)\}$. 
\end{proof}

\begin{corollary}
Let $\hat{f}, \hat{g} : |K| \xrightarrow{} \Rspace$ be simplex-wise linear functions, and let $\tilde{F}$ and $\tilde{G}$ be the sublevel set filtrations of $\hat{f}$ and $\hat{g}$ respectively. Then, $$d_{B}({\obarc(\tilde{F})}, {\obarc(\tilde{G}})) \leq || \hat{f}-\hat{g}||_{\infty}.$$
\end{corollary}

\begin{proof}

    
    Let $m$ be the number of simplices in $K$. Let $\eps = ||\hat{f}-\hat{g}||_\infty$. There is an $\eps$-interleaving between the filtration $\tilde{F}$ and $\tilde{G}$ defined by inclusion of sublevel sets~\cite{Cohen-SteinerEdelsbrunnerHarer2005,ChazalCohenSteinerGlisse2009}. We can choose $\delta$ such that $\eps \gg m\delta$.
    Then there is an $(\eps+m\delta)$-interleaving between $\tilde{F}(\delta)$ and $\tilde{G}(\delta)$, since $K_{t} \subset K'_{t+\eps+m\delta}$ and $K'_{t} \subset K_{t+\eps+m\delta}$, where $K_t$ and $K'_t$ are the complexes of $\tilde{F}(\delta)$ and $\tilde{G}(\delta)$ respectively. Moreover, by~\cref{lemma:barsnapping}, $\max\{\rho(F), \rho(G)\} \leq m\delta$. Therefore, for all $\delta$ such that $m\delta \ll \eps$, $d_B ( {\obarc(\tilde{F}(\delta))}, {\obarc(\tilde{G}(\delta))} )  \leq \eps+m\delta = ||f-g||_\infty+m\delta$. Since $\delta$ approaches 0 the corollary follows. 
\end{proof}

\section*{Acknowledgment}
T.\ Hou is supported by NSF fund CCF 2439255. B.\ Wang is supported by DOE DE-SC0021015, NSF IIS-2145499 and NSF DMS-2301361. 
\bibliography{refs-harmonic}
\appendix  
\section{Homology and Cohomology}
\label{sec:homology}

Let $K$ be a simplicial complex and we give the standard orientation to simplices in $K$. We write simplices as an ordered set of vertices, e.g.,~$\sigma=v_0v_1\cdots v_{p}$ for a $p$-dimensional simplex, for any integer $p \geq 0$.
Let $K_p$ denote the set of $p$-simplices of $K$ and $n_p=|K_p|$. 
The $p$-dimensional \emph{chain group} of $K$ with coefficients in $\Rspace$, denoted $C_p(K)$, is an $\Rspace$-vector space generated by $K_p$. 
By fixing an ordering of the set $K_p$, we can identify any $p$-chain $c\in C_p(K)$ with an ordered $n_p$-tuple with real entries (e.g.,~an element of $\Rspace^{n_p}$).
For each $p$, we fix an ordering for the $p$-simplices once and for all, and identify $C_p(K)$ and $\Rspace^{n_p}$. 
The standard basis of $\Rspace^{n_p}$ corresponds (under the identification) to the basis of $C_p(K)$ given by the simplices with standard orientation.

The $p$-dimensional boundary matrix $\partial_p: C_p(K) \rightarrow C_{p-1}(K)$ is defined on a simplex basis element by the formula
\[ \partial (v_0v_1 \cdots v_p)= \sum_{j=0}^{p} {(-1)^q (v_0\cdots v_{q-1}v_{q+1}\cdots v_p) }.\]
In the right hand side above, in the $q$-th term $v_q$ is dropped. The formula guarantees the crucial property of the boundary homomorphism: $\partial_{p} \partial_{p+1} = 0$. This simply means that the boundary of a simplex has no boundary. The sequence $C_p(K)$ together with the maps $\partial_p$ define the \textit{simplicial chain complex} of $K$ with real coefficients, denoted $C_{\bullet}(K)$. The group of $p$-dimensional \textit{cycles}, denoted $Z_p(K)$, is the kernel of $\partial_p$. The group of $p$-dimensional boundaries, denoted $B_p(K)$, is the image of $\partial_{p+1}$. The $p$-dimensional homology group of $K$, denoted $H_p(K)$, is the quotient group $Z_p(K)/B_p(K)$. 
As a set, this quotient is formally defined as $\{z+B_p(K) \mid  z\in Z_p(K)\}$. The operations are inherited from the chain group. In words, the homology group is obtained from $Z_p(K)$ by setting any two cycles which differ by a boundary to be equal. All these groups are $\Rspace$-vector spaces.

Consider the space $\Rspace^{n_p}$. The cycle group $Z_p(K) \subset C_p(K) = \Rspace^{n_p}$ is a subspace, that is, a hyperplane passing through the origin. Similarly, the boundary group $B_p(K)$ is a subspace, and is included in $Z_p(K)$. The homology group is the set of parallels of $B_p(K)$ inside $Z_p(K)$. Each such parallel hyperplane differs from $B_p(K)$ by a translation given by some cycle $z$. These parallel hyperplanes partition $Z_p(K)$. The homology group is then isomorphic to the subspace perpendicular to $B_p(K)$ inside $Z_p(K)$. 
The dimension of the $p$-dimensional homology group is called the $p$-th $\Rspace$-Betti number, denoted $\beta_p(K)$. 

Simplicial \textit{cohomology} with coefficients in $\Rspace$ is usually  defined by the process of dualizing. This means that we replace an $p$-chain by a linear functional $C_p(K) \rightarrow \Rspace$, called a $p$-dimensional \textit{cochain}. 
The set of all such linear functionals is a vector space isomorphic to $C_p(K)$, called the cochain group, denoted $C^p(K)$, which is a dual vector space of $C_p(K)$.  
For the purposes of defining harmonic chains, we must fix an isomorphism. 
We take the isomorphism that sends each standard basis element of $\Rspace^{n_p}$, corresponding to $\sigma$, to a functional that assigns 1 to $\sigma$ and 0 to other basis elements, denoted $\hat{\sigma} \in C^p(K)$. Any cochain $\gamma \in C^p(K)$ can be written as a linear combination of the $\hat{\sigma}$. Therefore, it is also a vector in $\Rspace^{n_p}$. The fixed isomorphism allows us to identify $C^p$ with $\Rspace^{n_p}$, and hence to $C_p$. Therefore, any vector in $\Rspace^{n_p}$ is at the same time a chain and a cochain.

The \textit{coboundary matrix} $\delta^p: C^p(K)\rightarrow C^{p+1}(K)$, in its matrix representation, is the transpose of $\partial_{p+1}$, $\delta^p = \partial_{p+1}^\top$. It follows that $\delta_{p+1}\delta_p=0$, and we can form a \textit{cochain complex} $C^{\bullet}(K)$. The group of \textit{$p$-cocycles}, denoted $Z^p$ is the kernel of $\delta^p$. The group of \textit{$p$-coboundary} is the image of $\delta_{p-1}$. The $p$-dimensional \textit{cohomology group}, denoted $H^p(K)$, is defined as $H^p(K) = Z^p(K)/B^p(K)$. All of these groups are again vector subspaces of $C^p(K)$ and thus of $\Rspace^{n_p}$.
It is a standard fact that homology and cohomology groups with real coefficients are isomorphic.
\newcommand{\img}{\mathrm{img}}

\section{Justification for the Definition of Harmonic Representatives}\label{sec:rep-just}

Consider a zigzag module $\Mcal:V_0\lrarrowsp{g_0}V_1\lrarrowsp{g_1}\cdots\lrarrowsp{g_{k-1}}V_k$
which is \emph{elementary} in the sense that each $g_i$
is either an isomorphism, an injection with corank $1$,
or a surjection with nullity $1$.
Representatives for intervals in $\barc(\Mcal)$
are defined as follows:

\begin{definition}[Zigzag representatives \cite{maria2015zigzag}]
\label{dfn:rep-cls-ori}
A \emph{representative} for $[b,d]\in\barc(\Mcal)$
is a sequence $\{v_i\in V_i\mid i\in[b,d]\}$
such that  for every $b\leq i<d$, 
either $g_i(v_i)=v_{i+1}$ or
$g_i(v_{i+1})=v_{i}$
based on the direction of $g_i$.
Furthermore, we have:
\begin{description}
    \item[Birth condition:] 
    If $g_{b-1}:V_{b-1}\to V_{b}$ is forward
    (thus being injective),
    $v_b\not\in\img(g_{b-1})$;
    if $g_{b-1}:V_{b-1}\leftarrow V_{b}$ is backward
    (thus being  surjective),
    then $v_b$ is the non-zero element in $\ker(g_{b-1})$.
    
    \item[Death condition:]
    If $d<k$ and
    $g_{d}:V_{d}\leftarrow V_{d+1}$ is backward
    (thus being injective),
    $v_d\not\in\img(g_{d})$;
    if $d<i$
    and $g_{d}:V_{d}\to V_{d+1}$ is forward
    (thus being surjective),
    then $v_d$ is the non-zero element in $\ker(g_{d})$.
\end{description}
\end{definition}

\begin{proof}[Proof of \Cref{prop:single-rep}]
Let $[b,d]$ be an interval in $\hbarc_*(\filt)=\barc(\har_*(\filt))$
and let $\{z_i\in\har_*(K_i)\mid i\in[b,d]\}$
be a sequence of representatives for $[b,d]$
as in \Cref{dfn:rep-cls-ori}.
For each $i<d$,
we have that $z_i$
maps to $z_{i+1}$ 
or $z_{i+1}$
maps to $z_{i}$ 
by the inclusion,
and hence $z_i=z_{i+1}$.
Therefore, the representative
$\{z_i\in\har_*(K_i)\mid i\in[b,d]\}$ indeed
consists of a single cycle.
\end{proof}

To see that \Cref{dfn:rep} is an adaption of \Cref{dfn:rep-cls-ori}
to the module $\har_p(\filt)$, we only need to check the birth
and death conditions.
Let $[b,d]$ be an interval in $\hbarc_p(\filt)$
with a representative $z$.
For the birth condition, since inclusion maps have trivial kernels,
we have that $z$ is not in the image of the inclusion $\har_p(K_{b-1})\incto \har_p(K_b)$
by \Cref{dfn:rep-cls-ori}.
This means that 
$z\in\har_p(K_{b})\setminus\har_p(K_{b-1})$ 
which is as stated in \Cref{dfn:rep}.
The death condition can also be similarly verified.

\section{Justification for $\Harmat^p$ Forming a Basis for $\har_p(K_i)$}\label{sec:harmat-basis}

We first have the following proposition which follows directly from definitions:
\begin{proposition}\label{prop:har-rep-basis}
For a $K_i$ in $\filt$,
let $\{[b_j,d_j]\mid j\in J\}$ 
be all the intervals in $\hbarc_p(\filt)$ containing $i$,
where each $[b_j,d_j]$ has a harmonic representative $z_j$.
Then $\{z_j\mid j\in J\}$ forms a basis for $\har_p(K_i)$.
\end{proposition}

Define a \emph{prefix} $\filt_i$ of $\filt$ as the filtration
$\filt_i: 
K_0 \incto K_1 \incto
\cdots \incto K_i.$
It is then not hard to see the following fact concerning \Cref{alg:frame}:
\begin{proposition}\label{prop:iter-unpair-to-i}
Before each iteration $i$
processing $K_{i}\incto K_{i+1}$,
$\{[b,i]\mid b\in U^p\}$
is the set of all intervals containing $i$
in $\hbarc_p(\filt_i)$
and each partial representative for $b$
maintained by the algorithm
is indeed a harmonic representative for $[b,i]\in \hbarc_p(\filt_i)$.
\end{proposition}

Combining \Cref{prop:har-rep-basis,prop:iter-unpair-to-i},
we have that columns in $\Harmat^p$ form a basis for $\har_p(K_i)$.
\end{document}